\documentclass[12pt, onecolumn]{IEEEtran}

\usepackage{cite}
\usepackage{pslatex}
\usepackage{epsfig}
\usepackage{amsmath}
\usepackage{amssymb}
\usepackage{booktabs}
\usepackage{array}
\usepackage{multicol}
\usepackage{color}
\usepackage{url}
\usepackage{dsfont}
\usepackage{threeparttable}
\psfull

\newtheorem{theorem}{Theorem}
\newtheorem{lemma}{Lemma}

\newtheorem{definition}{Definition}
\newtheorem{coro}{Corollary}
\newtheorem{proposition}{Proposition}

\DeclareMathAlphabet{\mathpzc}{OT1}{pzc}{m}{it}

\renewcommand{\baselinestretch}{1.6}
\IEEEoverridecommandlockouts 

\def\bE{{\mathbf E}}

\def\bH{{\mathbf H}}
\def\bI{{\mathbf I}}

\def\e{{\rm e}}

\def\E{{\mathbf E}}

\def\b0{{\mathbf 0}}

\begin{document}

\title{Clean relaying aided cognitive radio under the coexistence constraint}

\author{
\authorblockN{Pin-Hsun Lin, Shih-Chun Lin, Hsuan-Jung Su and Y.-W. Peter Hong}
\thanks{Pin-Hsun Lin and Hsuan-Jung Su  are with Department of
Electrical Engineering and Graduate Institute of Communication
Engineering, National Taiwan University, Taipei, Taiwan 10617.
Shih-Chun Lin and Yao-Win Peter Hong are with Institute of
Communications Engineering, National Tsing Hua University,
HsinChu, Taiwan, 30013. Emails: \{pinhsunlin@gmail.com,
linsc@mx.nthu.edu.tw, hjsu@cc.ee.ntu.edu.tw,
ywhong@ee.nthu.edu.tw\}. The material in this paper was presented
in part at the Annual Conference on Information Sciences and
Systems, 2010. This work was supported by the National Science
Council, Taiwan, R.O.C., under grant NSC 98-2219-E-002-016.} }
\maketitle \thispagestyle{empty} \vspace{-15mm}
{\renewcommand{\baselinestretch}{1.45}
\begin{abstract}
We consider the interference-mitigation based cognitive radio
where the primary and secondary users can coexist at the same time
and frequency bands, under the constraint that the rate of the
primary user (PU) must remain the same with a single-user decoder.
To meet such a coexistence constraint, the relaying from the
secondary user (SU) can help the PU's transmission under the
interference from the SU. However, the relayed signal in the known
dirty paper coding (DPC) based scheme is interfered by the SU's
signal, and is not ``clean''. In this paper, under the half-duplex
constraints, we propose two new transmission schemes aided by the
clean relaying from the SU's transmitter and receiver without
interference from the SU. We name them as the clean transmitter
relaying (CT) and clean transmitter-receiver relaying (CTR) aided
cognitive radio, respectively. The rate and multiplexing gain
performances of CT and CTR in fading channels with various
availabilities of the channel state information at the
transmitters (CSIT) are studied. Our CT generalizes the celebrated
DPC based scheme proposed previously. With full CSIT, the
multiplexing gain of the CTR is proved to be better (or no less)
than that of the previous DPC based schemes. This is because the
silent period for decoding the PU's messages for the DPC may not
be necessary in the CTR. With only the statistics of CSIT, we
further prove that the CTR outperforms the rate
performance of the previous scheme in fast Rayleigh fading channels. The numerical examples also show that in a large class of channels, the proposed CT and CTR provide significant rate gains over the previous scheme with small complexity penalties.\\
\end{abstract}
}
\vspace{-13mm}

\section{Introduction}
\vspace{-2mm} Efficient spectrum usage becomes a critical issue to
satisfy the increasing demands for high data rate services. Recent
measurements from the Federal Communications Commission (FCC) have
indicated that ninety percent of the time, many licensed frequency
bands remain unused and are wasted. Cognitive radio
\cite{Mitola_CR} is a promising technique to cope with such
problems by accessing the unused spectrum dynamically. This new
technology is capable of dynamically sensing and locating unused
spectrum segments in a target spectrum pool, and communicating via
the unused spectrum segments without causing harmful interference
to the primary users. The primary user (PU) is the user who
communicates in the licensed band using existing commercial
standards, while the user who uses the cognitive radio technology
is called the secondary user (SU). Originally, the cognitive radio
adopts the interference avoidance methodology, that is, if a PU
demands the licensed band, the SU should vacate and find an
alternative one. Recently, the concept of \textit{interference
mitigation} was proposed for the cognitive radio
\cite{Devroye_CR_achievable_rate}, where the SU and PU can coexist
and simultaneously transmit at the same time and frequency bands
to further improve the spectrum efficiency. The key is to allow
cooperations between the transmitters of the SU and PU. To make
the interference-mitigation based cognitive radio in
\cite{Devroye_CR_achievable_rate} more practical, the
\textit{coexistence constraint} was further proposed in
\cite{Viswanath_CR}. The cognitive radio is forced to maintain the
same PU rate performance as if it is silent, under the constraint
that the decoder of PU must be a \textit{single-user} decoder,
such as the conventional minimum distance decoder. Assuming that
the PU's message is known by the SU, in \cite{Viswanath_CR}, the
SU's transmitter not only transmits its own signal but also relays
the PU's signal to meet the coexistence constraint. Moreover, by
precoding with the celebrated dirty paper coding (DPC)
\cite{CostaDPC}, the SU's receiver can decode as if the
interference from the PU does not exist. Indeed, such a
transmission scheme is proved to be capacity-achieving in some
channel conditions \cite{Viswanath_CR}.

\setcounter{page}{1}

However, there are still some deficiencies and impractical
assumptions in the cognitive radio proposed in \cite{Viswanath_CR}
which motivate our work. First, in \cite{Viswanath_CR}, the
relayed PU's signal and the SU's own signal are simultaneously
transmitted. Since the SU's signal is an interference to the PU's
receiver, it pollutes the relaying and may cause power
inefficiency. Second, the DPC requires that the SU's transmitter
knows the PU's message. It may be hard to satisfy this
requirement, especially when the channel between the transmitters
of the PU and SU is not good enough. Finally, the perfect channel
state information at the transmitter (CSIT) may not always be
available, epically when the channel is fast faded. Without full
CSIT, the DPC used in \cite{Viswanath_CR} suffers
\cite{scCR_TWCOMs09}. To solve these problems, we propose two new
transmission schemes for cognitive radio which are aided by the
``clean'' relaying to the PR's receiver without the interference
from the SU. Under the half-duplex constraint, the clean relaying
comes from the transmitter or/and the receiver of SU, thus we name
the proposed schemes as the clean transmitter relaying (CT) and
the clean transmitter-receiver relaying (CTR) aided cognitive
radio, respectively.

Our main contributions are proposing the new CT and CTR to improve
the performance in \cite{Viswanath_CR}. Our CT generalizes the
DPC-precoded cognitive radio in \cite{Viswanath_CR}. Moreover, our
CTR can also avoid the last two problems mentioned in the previous
paragraph since it does not require the DPC. The cooperation
method of the CTR makes it face a multiple-access channel (MAC)
with common message, and we adopt the optimal signaling for this
channel from \cite{liu2006capacity} in the CTR. We also invoke the
channel coding theorem in \cite{Pombra_feedback_capacity} to
ensure that the coexistence constraint is met under the relaying.
With full CSIT and high signal-to-noise ratio (SNR), we find that
the multiplexing gain performance of the CTR is better than (or at
least no less than) that of \cite{Viswanath_CR}. This is due to
the fact that the silent period spent on decoding the PU's
messages for the DPC in \cite{Viswanath_CR} may not be necessary
in the CTR. When there is only the statistics of CSIT, the CTR is
even more promising. We observe that the DPC used in
\cite{Viswanath_CR} fails in fast Rayleigh fading channels, that
is, the rate performance of the SU is the same as that of treating
the interference from the PU as pure noise at the SU's receiver.
Then the CTR always has better rate performance than that of
\cite{Viswanath_CR} for all SNR regimes. We also identify the
structure of the optimal common message relaying ratio for the CTR
by exploring the corresponding stochastic rate optimization
problem. Simulation results verify the superiority of the proposed
CT and CTR over methods in \cite{Viswanath_CR} in terms of rates
and multiplexing gains under a large class of channels. Finally,
the complexity of the CTR is lower than that in
\cite{Viswanath_CR}, while the complexity of the CT is
approximately the same as that in \cite{Viswanath_CR}. The former
is because the signaling from \cite{liu2006capacity} adopted in
the CTR is much easier to implement in practice than the
complicated DPC \cite{DPC_Erez_IT05_RA}.

The cognitive channel model studied in the paper is related to
\cite{Wang_sIT09}\cite{ng2007capacity}, where cooperations in
interference channels were studied. However, the coexistence
constraints were not imposed in these papers, and thus the relay
strategies could be more flexible to obtain better rate
performance compared with ours. As noted in \cite{Viswanath_CR},
the capacity results for these less restricted channels can serve
as the performance outer bounds for our setting. Moreover, full
CSIT is usually assumed in the literatures
\cite{Devroye_CR_achievable_rate}\cite{Viswanath_CR}\cite{Wang_sIT09}\cite{ng2007capacity}
(also in our previous work \cite{scCR10}), while this work also
considers the partial CSIT case. With only the statistics of CSIT,
we show that our CTR outperforms the DPC based schemes in
\cite{Viswanath_CR}\cite{scCR_TWCOMs09} in fast Rayleigh fading
channels. In addition, the CT and multiplexing gain analysis also
are new, and did not appear in our previous works \cite{scCR10}.

The paper is organized as following. The system model is discussed
in Sec. \ref{Sec_system_model}. In Sec.\ref{sec_UT} and
\ref{sec_UTR}, we present the proposed CT and CTR and their rate
and multiplexing gain performances with full CSIT, respectively.
The performance analysis and the optimal common message relaying
ratio with only the statistics of CSIT in fast Rayleigh fading
channels are given in Sec. \ref{sec_ergodic}. We provide numerical
examples in Sec. \ref{Sec_Simu}. Finally, Sec.
\ref{Sec_Conclusion} concludes this paper.

\vspace{-3mm}
\section{System Model}\label{Sec_system_model}
\vspace{-3mm}
\subsection{Notations}
\vspace{-3mm} In this paper, the superscript $(.)^H$ denotes the
transpose complex conjugate. Identity matrix of dimension $n$ is
denoted by $\mathbf{I}_n$. A block-diagonal matrix with diagonal
entries $\mathbf{A}_1, \ldots, \mathbf{A}_k$ is denoted by
$diag(\mathbf{A}_1, \ldots, \mathbf{A}_k)$; while $|\mathbf{A}|$
and $|a|$ represent the determinant of a square matrix
$\mathbf{A}$ and the absolute value of a scalar variable $a$,
respectively. The mutual information between two random variables
is denoted by $I(;)$. We define $C(x) \triangleq \log(1+x)$ (the
base of $\log$ function is 2), and the function $(x)^+$ as
$(x)^+=x$ if $x\geq 0$, otherwise, $(x)^+=0$. Also the indicating
function $\mathbf{1}_A$ is one if the event $A$ is valid, and is
zero otherwise. \vspace{-5mm}
\subsection{Cognitive channel model}
\vspace{-2mm} As shown in Fig.~\ref{CR_channel}, in the considered
four-node cognitive channel, Node 1 and 2 are the transmitters of
PU and SU while Node 4 and 3 are the corresponding receivers,
respectively. For the $t$-th symbol time where $t$ is the discrete
time index, the received signals $Y_2(t)$, $Y_3(t)$ and $Y_4(t)$
at Node 2, 3 and 4 can be respectively represented by
\begin{equation} \label{Eq_channel}
\left[\begin{array}{c} Y_4(t)\\
Y_3(t) \\ Y_2(t) \end{array}\right]=\left[\begin{array}{ccc} h_{14}(t) &h_{24}(t)&h_{34}(t)\\
h_{13}(t) &h_{23}(t)&0 \\ h_{12}(t) &0 &0 \end{array}\right]\left[\begin{array}{c}X_1(t)\\X_2(t) \\X_3(t)\end{array}\right]+\left[\begin{array}{c} Z_4(t)\\
Z_3(t) \\
Z_2(t)\end{array}\right],
\end{equation}
where the channel gain between node $i$ and $j$ is denoted by
$h_{ij}(t)$, and $Z_i(t)$ is the additive white Gaussian noise
process at node $i$. Each time sample of $Z_i(t)$ is independent
and identically distributed (i.i.d.) circularly-symmetric complex
Gaussian, i.e., $Z_i \sim \mathcal{CN}(0,1)$. Signals transmitted
from Node 1, 2 and 3 are denoted as $X_1(t)$, $X_2(t)$, and
$X_3(t)$ with long term average power constraints $\bar{P}_1$,
$\bar{P}_2$, and $\bar{P}_3$, respectively as
\begin{equation} \label{Eq_avg_power}
\frac{1}{n}\sum^n_{t=1}[|X_i(t)|^2] \leq \bar{P}_i,
\;\;\mathrm{for}\;\;i=1,\,2,\,3,
\end{equation}
where $n$ is the number of coded symbols in a codeword. Note that
all nodes are \textit{half-duplex}.

In this paper, we consider two cases with different channel
knowledge of $h_{ij}(t)$ at the transmitter, while the channel
gains $h_{ij}(t)$ are always assumed perfectly known at the
corresponding receivers. In the first case,
$h_{ij}(t)=h_{ij}=|h_{ij}|e^{j \theta_{ij}}, \;\; \forall 1 \leq t
\leq n$, where $\theta_{ij}$ is the channel phase. As for the CSIT
assumptions, we assume that Node 1 knows $h_{14}$, Node 3 knows
$h_{34}$, and Node 2 knows all channel gains based on the method
proposed in \cite{Viswanath_CR}. The second case is the fast
Rayleigh fading channel, where each $h_{ij}(t)$ is varying at each
$t$. We assume that $h_{ij}(t)$ are i.i.d. generated according to
a random variable $H_{ij}$, and $H_{ij}$ is complex Gaussian
distributed with zero mean and variance $\sigma^2_{ij}$. Moreover,
due to the limited channel feedback bandwidth, we assume that the
channel realizations $h_{ij}(t)$ are unknown at the transmitters.
However, Node 1 knows the statistics of $H_{14}$, Node 3 knows the
statistics of $H_{34}$, and Node 2 knows the statistics of all
channels by applying the methods in
\cite{scCR_TWCOMs09}\cite{Viswanath_CR}. The SU also knows the
target rate of the PU by using the methods in \cite[Sec.
II]{scCR_TWCOMs09}.

We restrict the decoder of PU at Node 4 as a single-user decoder.
A single-user decoder $D_s$ is defined to be any decoder which
performs well on the point-to-point channel with perfect channel
state knowledge at the decoder \cite{Viswanath_CR}. Without loss
of generality, we set the decoder to be the maximum-likelihood
decoder for fading channel with temporal independent Gaussian
noise as in \cite{viterbo2002universal} (minimum-distance
decoder). We then define the achievable rate under such decoder as
the following.
\begin{definition} \label{Def_SU}
A rate $R_1$ is \textit{single-user achievable} for the PU if
there exists a sequence of $(2^{nR_1},n)$ encoders $E^n_1$ that
encodes PU's message $w_1$, such that the average probability of
error vanishes to zero as $n \rightarrow \infty$ when the receiver
uses a single user decoder $D_s$.
\end{definition}
Denote the set of all primary encoders that map primary messages
to the transmitted signals as $E^n_1$ , we then have the following
definition.
\begin{definition}\label{Def_CR_code}
A cognitive radio code with rate $R_2$ and length $n$ consists of
an encoder to encode the SU's message $w_2$ with output
$X^n_2=\{X_2(1),\ldots,X_2(n)\}$ as $ E^n_2 : E^n_1 \times
\{1,\ldots,2^{nR_1}\} \times \{1,\ldots,2^{nR_2}\} \rightarrow
X^n_2, $ where $\|X^n_2\|^2/n \leq \bar{P}_2$, and a decoder to
decode message $w_2$ from the received signal
$Y^n_3=\{Y_3(1),\ldots,Y_3(n)\}$.
\end{definition}

Based on Definition \ref{Def_CR_code}, we have the following
definition for the achievable rate of the cognitive radio under
the \textit{coexistence constraint} \cite{Viswanath_CR} .

\begin{definition} \label{Def_CR}
The \textit{coexistence constraint} means that for a given PU's
rate $R_T$, the SU must take $R_T$ as a rate target and ensure
that under its own transmissions, $R_T$ is still
\textit{single-user achievable} for the PU as defined in
Definition \ref{Def_SU}. A rate $R_2$ is achievable for the SU if
there exists a sequence of $(2^{nR_2},n)$ cognitive radio codes
defined in Definition \ref{Def_CR_code} such that under the
coexistence constraint, the average probability of error vanishes
to zero  as $n \rightarrow \infty$.
\end{definition}

\vspace{-4mm}
\section{Clean Transmitter Relaying in Channels with Full CSIT} \label{sec_UT}
\vspace{-1mm}

For simplicity, we will introduce the CT aided cognitive radio and
its performance in channels with full CSIT first. Then we will
discuss the CTR which further allows the relaying from the SU's
receiver in Section \ref{sec_UTR}. As shown in Fig.~\ref{Fig_UTR}
(a), the new three-phase CT is a generalization of the two-phase
cognitive radio in \cite{Viswanath_CR}, by introducing an
additional ``clean'' relay link (without interference from the SU)
from Node 2 in the third phase. As will be shown later, to know
the PU's message $w_1$ for the DPC operation, Node 2 needs Phase 1
to be long enough to correctly decode $w_1$ from the received
signal. However, this phase is neglected in \cite{Viswanath_CR}
and most of the existing works. It is also clear from
Fig.~\ref{Fig_UTR} (a) that due to the half-duplex constraint, the
transmission scheme must be multi-phase since Node 2 cannot
receive and transmit at the same time. As will be clarified later,
the multi-phase transmission will cause SNR changes at Node 4 in
different phases. To deal with this new problem, we need to invoke
the upcoming Lemma \ref{Lemma_WeiYu} to meet the coexistence
constraint.

The detailed CT signaling method of each phase in
Fig.~\ref{Fig_UTR} (a) comes as the following. To simplify the
notations, we omit the time index of the signals in
\eqref{Eq_channel} to represent the corresponding signals in the
Shannon random coding setting \cite{Book_Cover}. For example,
$X_1$ corresponds to $X_1(t)$. Assume that each three-phase
transmission occupies $n$ symbol times, which forms a codeword. We
have

\noindent \textbf{Phase 1}: Within the first $\lfloor t_1n \rfloor
$ symbols, Node 2 listens to and decodes the PU's message $w_1$.
Here $t_i$ is the portion of time of Phase $i$, $i=1,2,3$.

\noindent \textbf{Phase 2}: If the decoding of $w_1$ is
successful, within the next $\lfloor t_2n \rfloor$ symbols, Node 2
sends the DPC encoded signal $X^D_2$ using side-information
$h_{13}X_1$ plus the relaying of $X_1$ as
\begin{equation} \label{Eq_Vish_X2}
    X_2=X^D_2+\sqrt{\alpha_1\frac{P_2}{P_1}}e^{j(\theta_{14}-\theta_{24})}X_1,
\end{equation}
where message $w_2$ is conveyed in $X^D_2$, $\alpha_1$ is the
relay ratio from Node 2 to maintain the rate performance $R_T$ of
the primary link, while $P_1$ and $P_2$ are the transmitted power
of $X_1$ and $X_2$ in Phase 2, respectively.

\noindent \textbf{Phase 3}: For the remaining $ t_3n=n-\lfloor
t_1n \rfloor-\lfloor t_2n \rfloor $ symbols, the clean relaying is
transmitted from Node 2 to assist decoding at Node 4 as
\begin{equation} \label{Eq_UT_X2_Phase2}
X_2= \sqrt{P_2/P_1}\e^{j(\theta_{14}-\theta_{24})}X_1.
\end{equation}
To meet average power constraints \eqref{Eq_avg_power}, the power
$P_1$ and $P_2$ are set as
\begin{equation} \label{Eq_UT_power}
P_1=\bar{P}_1 \;\; \mbox{and} \;\; (1-t_1)P_2=\bar{P}_2.
\end{equation}
After Phase 1, Node 2 knows the PU's message $w_1$. Since Node 2
also knows the PU's codebook from Definition \ref{Def_CR_code},
the PU's transmitted codeword and then the interference at Node 3
$h_{13}X_1$ is known at Node 2. The DPC results in \cite{CostaDPC}
can be applied by using $h_{13}X_1$ as the non-causally known
transmitter side-information which is unknown at Node 3.

Note that the received SNR of $X_1$ at Node 4 changes in different
phases (different block of symbols). To meet the coexistence
constraint in Definition \ref{Def_CR} under this phenomena, we
introduce a Lemma as
\begin{lemma} \label{Lemma_WeiYu}
For a block of $n$ transmissions over the channel $Y^n=\bH^n
X^n+Z^n$, where $n \times 1$ vectors $X^n$ and $Y^n$ are the
transmit and received signals respectively, the diagonal channel
matrix $\bH^n$ is known at the receiver, and $Z^n$ is a Gaussian
random sequence with diagonal covariance matrix $K_{Z^n}$ (each
element of $Z^n$ may not be identically distributed), the coding
rate $R$ is \textit{single-user achievable} for Gaussian codebooks
if
\begin{equation} \label{eq_WeiYu}
R<\frac{1}{n}\log\frac{|\bH^n
K_{X^n}(\bH^n)^H+K_{Z^n}|}{|K_{Z^n}|},
\end{equation}
where the covariance matrix $K_{X^n}$ of the transmitted signal
$X^n$ satisfies the power constraint.
\end{lemma}

In Lemma \ref{Lemma_WeiYu}, the $n \times n$ channel matrix
$\bH^n$ is a collection of scalar time-domain channel coefficients
over the $n$ transmissions. It is different to the spatial-domain
channel matrix over single transmission in the multiple-antenna
system \cite{zheng_tse_2003}, where the vector channels are
assumed to be i.i.d in time. The proof of this lemma follows the
steps in \cite{Pombra_feedback_capacity} where the asymptotic
equipartition property for arbitrary Gaussian process is invoked
to prove that the right-hand-side (RHS) of \eqref{eq_WeiYu} is
achievable by the suboptimal jointly typical decoder. Then it is
also achievable by the optimal maximum-likelihood decoder defined
in Definition \ref{Def_SU}. The detail is omitted. Then we have
the following achievable rate result for the CT with the proof
given in Appendix \ref{App_CT}
\begin{theorem} \label{Theorem_UT}
With full CSIT and the transmitted power setting in
\eqref{Eq_UT_power}, the following rate of SU is achievable by the
CT \vspace{-4mm}
\begin{equation} \label{eq_UT_CR}
R_2 \leq \max_{t_2,\alpha_1} \;\; t_2 C(|h_{23}|^2
(1-\alpha_1)P_2),
\end{equation}
which is subject to the constraint for coexistence with
\begin{align} \label{eq_UT_coex}
 R_T<t_1
C(|h_{14}|^2P_1)+t_2C\left(\frac{(|h_{14}|\sqrt{P_1}+|h_{24}|\sqrt{\alpha_1P_2})^2}{1+|h_{24}|^2(1-\alpha_1)P_2}\right)
+(1-t_1-t_2)C\left(\left(|h_{14}|\sqrt{P_1}+|h_{24}|\sqrt{P_2}\right)^2\right),
\end{align}
and the constraint for Node 2 to successfully decode PU's message
as
\begin{equation} \label{eq_UT_DF}
t_1>R_T/C(|h_{12}|^2P_1),
\end{equation}
where the intervals of Phase 1 and 2 are $\lfloor t_2n \rfloor$
and $\lfloor t_2n \rfloor$, respectively, and the relaying ratio
$\alpha_1\in [0 ,1]$.
\end{theorem}

Note that in \cite{Viswanath_CR}, there is no Phase 3 ($t_3=0$)
and the relaying from SU is always ``noisy'' (interfered by SU's
own signal $X^D_2$) from \eqref{Eq_Vish_X2}. When the channel gain
$|h_{24}|$ is large, much of the SU's available power are used to
overcome the interference from SU's own signal and the
transmission may not be efficient ($\alpha_1$ is high). Also in
\cite{Viswanath_CR}, the assumption $|h_{12}| \gg |h_{14}|$ or
$I(X_1;Y_2) \gg R_T$ is made to ensure that $t_1$ can be
essentially neglected ($t_2=1$), and the SNR is almost the same
within a codeword. Thus the conventional Shannon channel coding
theorem in \cite{Book_Cover} can be invoked to ensure the
coexistence constraint. However, we usually have $t_1 \neq 0$ for
any more reasonable channel setting. Then the SNRs of Phase 1 and
2 are different at Node 4, that is, SNR changes in different block
of symbols. The Lemma \ref{Lemma_WeiYu}, which is more general
than that in \cite{Book_Cover}, is required to ensure that the
PU's rate is single-user achievable in Definition \ref{Def_SU}.
The insight of this Lemma is that even when Node 4 encounters bad
SNR at Phase 1, we can boost the SNR in Phase 2 to make the rate
over all phases unchanged. Note that since the equivalent channel
and noise change in different phases in the CT (also in its
special case where practical $t_1 \neq 0$ is introduced in
\cite{Viswanath_CR}), the decoder at Node 4 needs to be able to
track them. This can be done by the well-developed channel
estimation techniques in
\cite{tugnait2002single}\cite{otnes2004iterative}.

The optimization problem in \eqref{eq_UT_CR} is not convex in
$(t_2,\alpha_1)$ and the analytical solution is hard to obtain.
However, for fixed $t_2$, one can easily show that the optimal
$\alpha_1$ (function of $t_2$) is
\begin{equation} \label{eq_UT_alpha}
\alpha^*_1(t_2)=\left(\frac{-|h_{14}|\sqrt{P_1}+\sqrt{\left(1-|h_{14}|^2P_1+|h_{24}|^2P_2\right)K(t_2)+(1+|h_{24}|^2P_2)K^2(t_2)}}{|h_{24}|\sqrt{P_2}(1+K(t_2))}\right)^2,
\end{equation}
where $ K(t_2)=2^{\frac{1}{t_2}\left(
R_T-t_1C(|h_{14}|^2P_1)-(1-t_1-t_2)C(|h_{14}|^2P_1+|h_{34}|^2P_2)\right)}-1.
$ Note that if $t_1=t_3=0$, $K(1)=P_1$, then $\alpha^*_1(1)$ from
\eqref{eq_UT_alpha} equals to the one derived in
\cite{Viswanath_CR}. Since $0<t_2 \leq 1-t_1$, it is easy to find
the optimal $t_2$ maximizing \eqref{eq_UT_CR} by line search.

Now we study the multiplexing gain (or the pre-log factor)
\cite{zheng_tse_2003} of the SU, which is defined by
\begin{equation} \label{Eq_mul_def}
m_2=\lim_{\stackrel {\;}{\bar{P}_c} \rightarrow \infty}
R_2/\log\bar{P}_c,
\end{equation}
where $\bar{P}_c$ is the average transmission power utilized by
the SU. For the CT, $\bar{P}_c=\bar{P}_2$. The reason for
introducing $\bar{P}_c$ is to fairly compare the performance of CT
and CTR, of which the $\bar{P}_c$ is defined in the upcoming
\eqref{Eq_sum_power}. We focus only on the multiplexing gain of
the SU since that of the PU is unchanged with and without the
existence of SU due to the coexistence constraint. With
\eqref{eq_UT_CR} and \eqref{Eq_UT_power}, the upper bound of the
multiplexing gain of the CT can be easily found as \vspace{-2mm}
\begin{equation} \label{Eq_mul_CT}
m_2\leq 1-t_1= \left(
1-\frac{C(|h_{14}|^2\bar{P}_1)}{C(|h_{12}|^2\stackrel
{\;}{\bar{P}_1})}\right)^+.
\end{equation}
That is, the multiplexing gain is limited by the decoding time of
Phase 1, which is small when $|h_{12}|/|h_{14}|$ is small. This
motivates us to develop the CTR discussed in the next section.
\vspace{-3mm}
\section{Clean Transmitter-Receiver Relaying in Channels with full CSIT} \label{sec_UTR}
\vspace{-1mm} Although the proposed CT is more practical and
expected to outperform the cognitive radio in \cite{Viswanath_CR}
when $|h_{24}|$ is large, there are still some disadvantages.
First, when $|h_{12}|/|h_{14}|$ is small due to a deep fade from
Node 1 to Node 2 or a blockage in this signal path, from
\eqref{Eq_mul_CT}, the CT may fail since $t_1$ approaches 1. In
addition, the complexity of practical DPC implementation
\cite{DPC_Erez_IT05_RA} may still be inhibitive in current
communication systems. These problems motivate us to include the
clean relaying from Node 3, the SU's receiver, and develop the CTR
aided cognitive radio. We will show that with full CSIT, the
multiplexing gain of the CTR is no less than that of the CT (also
the special case \cite{Viswanath_CR}) with lower implementation
complexity. With only the statistics of the CSIT, the rate
performance of CTR is even more promising for fast Rayleigh faded
channels, as will be shown in Section \ref{sec_ergodic}. Again,
due to the half duplex constraint, the CTR transmission is
multi-phase since Node 2 and 3 cannot transmit/receive at the same
time.

The equivalent channel of each phase in the proposed CTR is
depicted in Fig.~\ref{Fig_UTR} (b). The basic design concept comes
as follows. After Phase 1, the PU's message $w_1$ is known by the
SU, and can be treated as a \textit{common message} for the PU and
SU. Thus in Phase 2, Node 3 faces an \textit{asymmetric} MAC with
a common message \cite{liu2006capacity}, since Node 3 also needs
to decode $w_1$ to enable clean relaying in Phase 3. Here the word
``asymmetric'' comes from the fact that the PU in this two-user
MAC can only transmit the common message $w_1$. The signaling
method (upcoming \eqref{Eq_UTR_X2}) in Phase 2 is then inspired
from the optimal signaling proposed in \cite{liu2006capacity}. Two
independent codebooks are used to transmit the private and common
messages $w_2$ and $w_1$ from Node 2, respectively. Note that we
can not use the signaling designed for the conventional
interference channels without the coexistence constraint such as
\cite{Wang_sIT09}, where the PU's receiver needs to decode part of
the SU's messages to get good rate performance. The detailed
signaling method for each phase comes as the following.

\noindent \textbf{Phase 1}: In the first $\lfloor t_1n \rfloor $
symbols, Node 2 and 3 listen to the PU's message $w_1$. Node 2
decodes $w_1$.

\noindent \textbf{Phase 2}: Within the next $\lfloor t_2n \rfloor$
symbols, Node 2 transmits
\begin{equation} \label{Eq_UTR_X2}
    X_2=U_2+\sqrt{\alpha_1\frac{P_2}{P_1}}\e^{j\theta_1}X_1,
\end{equation}
where $U_2$ is the signal bearing SU's message $w_2$ and is
independent of $X_1$, while $\alpha_1$ and $\theta_1$ are the
relaying ratio and phase for the common message $w_1$,
respectively. Node 3 decodes both $w_1$ and $w_2$.

\noindent \textbf{Phase 3}: For the remaining $ t_3n=n-\lfloor
t_1n \rfloor-\lfloor t_2n \rfloor $ symbols, the clean relaying
signals are transmitted from Node 2 and 3 as \vspace{-3mm}
\begin{equation} \label{Eq_UTR_X2_Phase2_2}
X_2= \sqrt{P^{(3)}_2/P_1}\e^{j(\theta_{14}-\theta_{24})}X_1, \;\;
X_3= \sqrt{P_3/P_1} \e^{j(\theta_{14}-\theta_{34})}X_1,
\end{equation}
respectively, where $P^{(3)}_2$ and $P_3$ are the transmitted
power of Node 2 and 3 at Phase 3 respectively. \\ To satisfy the
power constraint \eqref{Eq_avg_power}, we have \vspace{-3mm}
\begin{equation} \label{Eq_UTR_power}
P_1=\bar{P}_1, \;\; t_2P_2+t_3P^{(3)}_2=\bar{P}_2,\;\;\mbox{and}
\;\; t_3P_3=\bar{P}_3.
\end{equation}
It was shown in \cite{liu2006capacity} that other than the
complicated scheme in \cite{slepian1973coding}, the simple
signaling \eqref{Eq_UTR_X2} is also optimal for the Gaussian MAC
with common message. The low complexity advantage of our CTR is
then inherited from \cite{liu2006capacity}.

To calculate the achievable rate of the CTR, first note that the
received SNRs of $X_1$ and $U_2$ at Node 3 both change at Phase 1
and 2. Then we need the following Lemma from \cite{pombra1994non}.
Although Lemma \ref{Lemma_WeiYuMAC} is an extension of the
achievable rate in Lemma \ref{Lemma_WeiYu} to the MAC setting, in
Lemma \ref{Lemma_WeiYu}, we need to further prove that the rate is
\textit{single-user achievable} to meet the coexistence constraint
in Node 4. However, such requirement is not needed for Node 3
where Lemma \ref{Lemma_WeiYuMAC} is applied.

\begin{lemma} \label{Lemma_WeiYuMAC}
For a block of $n$ transmissions over the MAC $Y^n=\bH^n_x
X^n+\bH^n_uU^n+Z^n$, where the channel matrices $\bH^n_x$ and
$\bH^n_u$ are diagonal and known perfectly at the receiver, and
$Z^n$ is a Gaussian random sequence with covariance matrix
$K_{Z^n}$, the rate pair $(R_1,R_2)$ is achievable for Gaussian
codebooks if
\begin{align}
&R_1 \leq \frac{1}{n}\log\frac{|\bH^n_x K_{x^n}(\bH^n_x)^{H}+K_{z^n}|}{|K_{z^n}|}, \label{eq_PomMAC_R1}\\
R_2 \leq \frac{1}{n}\log\frac{|\bH^n_u
K_{u^n}(\bH^n_u)^{H}+K_{z^n}|}{|K_{z^n}|}&, \; R_1+R_2 \leq
\frac{1}{n}\log\frac{|\bH^n_x K_{x^n}(\bH^n_x)^{H}+\bH^n_u
K_{u^n}(\bH^n_u)^{H}+K_{z^n}|}{|K_{z^n}|}, \label{eq_PomMAC_R1R2}
\end{align}
where the covariance matrices $K_{x^n}$ and $K_{u^n}$ of the
transmitted signals $X^n$ and $U^n$ satisfy the power constraints,
respectively.
\end{lemma}

By combining the results in \cite{liu2006capacity} and Lemma
\ref{Lemma_WeiYuMAC} as well as using Lemma \ref{Lemma_WeiYu}, we
can choose PU's and SU's codebooks which can simultaneously ensure
successful decoding at Node 3, and meet the coexistence constraint
at Node 4. We then have the following achievable rate of the CTR
in Theorem \ref{Theorem_UTR}. Here the rate $R'_T$ can be treated
as, after Phase 1, the residual information flow of $w_1$ to be
decoded at Node 3; while $u_{T \!\! x}$ and $u_{Rx}$ in the end of
the theorem statement indicate whether the relaying from
transmitter and receiver are possible, respectively.

\begin{theorem} \label{Theorem_UTR}
With full CSIT and the transmitted power setting as
\eqref{Eq_UTR_power}, the following rate of the SU is achievable
by the CTR \vspace{-3mm}
\begin{align} \label{eq_UTR_CR}
R_2 \leq \max_{\theta_1,\,t_2,\,\alpha_1} \Big \{ \min \Big [ \;\; & t_2 \cdot C\left(|h_{23}|^2(1-\alpha_1)P_2\right), \notag\\
 &t_2 \cdot C \left(\big|h_{23}\sqrt{\alpha_1P_2/P_1}\e^{j
\theta_1}+h_{13}\big|^2P_1+|h_{23}|^2(1-\alpha_1)P_2\right)-R^m_T
\Big ] \Big \},
\end{align}
where $R^m_T=\min\left\{ R'_T,\,t_2 \cdot C
\left(\big|h_{23}\sqrt{\alpha_1P_2/P_1}\e^{j
\theta_1}+h_{13}\big|^2P_1\right)\right\}$ with $R'_T \triangleq
R_T-t_1C(|h_{13}|^2P_1)$, and is subject to the constraint for
coexistence with \vspace{-2mm}
\begin{align} \label{eq_UTR_coex}
R_T\leq\, &t_1 \cdot
C(|h_{14}|^2P_1)+t_2 \cdot C\left(\frac{\big|h_{14}+h_{24}\sqrt{\alpha_1P_2/P_1}\e^{j\theta_1}\big|^2P_1}{1+|h_{24}|^2(1-\alpha_1)P_2}\right)\notag\\
&+(1-t_1-t_2)C\left(\left(|h_{14}|\sqrt{P_1}+|h_{24}|\sqrt{P^{(3)}_2}+|h_{34}|\sqrt{P_3}\right)^2\right),
\end{align}
where the common message relaying ratio and phase are $\alpha_1$
and $\theta_1$, and the time fractions of Phase 1 and 2 are $t_1$
and $t_2$, respectively. Moreover, let $u_{T \!\!
x}=\mathbf{1}_{t_1>R_T/C(|h_{12}|^2P_1)}$, and
$u_{Rx}=\mathbf{1}_{t_2C(|h_{23}\sqrt{\alpha_1P_2/P_1}\e^{j
\theta_1}+h_{13}|^2P_1)\geq R_T'}$, $t_1$, $\alpha_1$, and
$P^{(3)}_2$ are all zero when $u_{T \!\! x}=0$, while $P_3=0$ when
$u_{Rx}=0$.
\end{theorem}

\begin{proof}
We first consider the case where $u_{T \!\! x}=1$ and $u_{Rx}=1$.
In this case, both Node 2 and 3 are capable of relaying with
$\alpha_1 \geq 0$, $P^{(3)}_2 \geq 0$ and $P_3 \geq 0$. As
explained in the beginning of Section \ref{sec_UTR}, Node 3 faces
an asymmetric MAC with common message $w_1$ and private message of
SU $w_2$. From \cite{liu2006capacity}, we know that one should
choose $X_1$ and $U_2$ independent and Gaussian distributed with
variance $P_1$ and $(1-\alpha_1)P_2$, respectively. The codebooks
of PU and SU are generated according to $X_1$ and $U_2$ with rate
$R_T$ and $R_2$, respectively. As suggested in
\cite{liu2006capacity}, the equivalent channel at Node 3 is
similar to a common two-user MAC without common message as in
\cite{Book_Cover}. However, as explained previously, the
difference between this MAC and that in \cite{Book_Cover} is that
both the SNRs of $X_1$ and $U_2$ at Node 3 vary during Phase 1 and
2. Then we need Lemma \ref{Lemma_WeiYuMAC} which is more general
than \cite{Book_Cover} to ensure correct decoding, with
$K_{u^n}=(1-\alpha_1)P_2\mathbf{I}_{n}$,$\;K_{x^n}=P_1\mathbf{I}_{n}$,
\[
\bH^n_u=\mathrm{diag}\left(0 \cdot \mathbf{I}_{\lfloor t_1n
\rfloor}, h_{23} \mathbf{I}_{\lfloor t_2n \rfloor},0 \cdot
\mathbf{I}_{\lfloor t_3n \rfloor} \right), \;
\bH^n_x=\mathrm{diag}\left(h_{13}\mathbf{I}_{\lfloor t_1n
\rfloor}, \; \left(h_{23}\sqrt{\alpha_1P_2/P_1}\e^{j
\theta_1}+h_{13}\right) \mathbf{I}_{\lfloor t_2n \rfloor},0 \cdot
\mathbf{I}_{\lfloor t_3n \rfloor} \right),
\]
and $K_{z^n}=\mathbf{I}_n,$ where \eqref{Eq_UTR_X2} in Phase 2 and
the channel model in \eqref{Eq_channel} are used. Then from
\eqref{eq_PomMAC_R1R2} in Lemma \ref{Lemma_WeiYuMAC}, the
following rate constraints apply for the correctly decoding of
$(w_1,w_2)$ at Node 3 in Phase 2 \vspace{-2mm}
\begin{align}
R_2& \leq t_2 \cdot C\left(|h_{23}|^2(1-\alpha_1)P_2\right), \notag \\
R_2+R_T &\leq t_1\cdot C(|h_{13}|^2P_1)+t_2 \cdot C
\left(|h_{23}\sqrt{\alpha_1P_2/P_1}\e^{j
\theta_1}+h_{13}|^2P_1+|h_{23}|^2(1-\alpha_1)P_2\right).
\label{eq_UTR_MAC_Cos}
\end{align}
With the above two inequalities, we have \eqref{eq_UTR_CR} with
$R_T^m=R_T-t_1C(|h_{13}|^2P_1)$. Since $u_{Rx}=1$,
$R_T^m=R_T'=R_T-t_1C(|h_{13}|^2P_1)$ by construction. Similarly,
with $u_{Rx}=1$ or $t_2C(|h_{23}\sqrt{\alpha_1P_2/P_1}\e^{j
\theta_1}+h_{13}|^2P_1)\geq R_T'$, inequality \eqref{eq_PomMAC_R1}
in Lemma \ref{Lemma_WeiYuMAC} is met by applying the above
procedure. The decoding of $(w_1,w_2)$ will then be successful. To
ensure the coexistence, by invoking Lemma \ref{Lemma_WeiYu},
\eqref{Eq_UTR_X2}, \eqref{Eq_UTR_X2_Phase2_2}, and
\eqref{Eq_channel}, and following the steps in Appendix
\ref{App_CT}, one can obtain \eqref{eq_UTR_coex}.

Now we consider the case $u_{T \!\! x}=1$ and $u_{Rx}=0$. It
happens when $R_T$ is too large for the MAC decoder in Node 3 to
successfully decode $w_1$. Then there is only relaying from Node 2
and no relaying from Node 3 at Phase 3 ($P_3=0$). With $X_1$ and
$U_2$ as described previously, Node 3 treats the PU's signal $X_1$
as pure Gaussian noise when decoding $w_2$. The achievable rate
$R_2$ is then
\begin{equation}\label{EQ_CR_treat_interference_as_noise}
t_2 \cdot C
\left(\frac{|h_{23}|^2(1-\alpha_1)P_2}{1+\big|h_{23}\sqrt{\alpha_1P_2/P_1}\e^{j
\theta_1}+h_{13}\big|^2P_1}\right).
\end{equation}
Note that \eqref{EQ_CR_treat_interference_as_noise} can be
rearranged as the second argument of the minimum in
\eqref{eq_UTR_CR} with \vspace{-2mm}
\begin{equation}\label{EQ_UR0_R1C}
R_T^m=t_2C(\big|h_{23}\sqrt{\alpha_1P_2/P_1}\e^{j
\theta_1}+h_{13}\big|^2P_1).
\end{equation}
When $u_{Rx}=0$, our definition of $R_T^m$ in the Theorem
statement will make \eqref{EQ_UR0_R1C} valid. Also the minimum in
\eqref{eq_UTR_CR} always equals to
\eqref{EQ_CR_treat_interference_as_noise} since $t_2 \cdot C
\left(|h_{23}|^2(1-\alpha_1)P_2\right)$ is always larger than
\eqref{EQ_CR_treat_interference_as_noise}, and
\eqref{EQ_CR_treat_interference_as_noise} equals to the second
argument in the minimum of \eqref{eq_UTR_CR} with
\eqref{EQ_UR0_R1C}.

Finally, we consider $u_{T \!\! x}=0$ and $u_{Rx}=1$, which
results in $\alpha_1=t_1=P^{(3)}_2=0$ and Node 2 cannot relay
$X_1$. However, as long as the clean relaying from Node 3 can
satisfy the coexistence
 constraint with $P_3>0$, the SU still can have non-zero rate.
Now Node 3 faces a conventional MAC channel without common message
and varying SNRs as in \cite{Book_Cover}. The analysis for $u_{T
\!\! x}=u_{Rx}=1$ includes this case as a special case, and
(\ref{eq_UTR_coex}) and (\ref{eq_UTR_CR}) are also valid. As for
cases where $u_{T \!\! x}=u_{Rx}=0$, $P_2$ must be zero to satisfy
(\ref{eq_UTR_coex}) since there is no relaying
$\alpha_1=P^{(3)}_2=P_2=0$. The SU's rate is zero from
(\ref{eq_UTR_CR}), and this concludes the proof.
\end{proof}
\noindent The optimization problem in Theorem \ref{Theorem_UTR} is
non-convex even
    when $t_2$ is given. However, since all variables are bounded, the complexity
     of numerical line search is still acceptable.

Note that our CTR uses different coding scheme compared with the
CT, and does not always guarantee rate advantage over CT under
full CSIT assumption. However, unlike the CT, even if
\eqref{eq_UT_DF} is violated and $u_{T \!\! x}=0$, the CTR may
still meet the coexistence constraint with only the relaying from
Node 3 ($u_{Rx}=1$). Even when $u_{T \!\! x}=1$, if Node 2 needs
too much time to decode $w_1$, setting $t_1=0$ in CTR (pure
receiver relaying) may has rate advantage over the CT. This
observation is verified in the upcoming high SNR analysis, where
the multiplexing gain of the CTR is shown to be larger than that
of the CT. In this analysis, the $\bar{P}_c$ in \eqref{Eq_mul_def}
equals to the sum of the average transmitted power from Node 2 and
3 (or total energy consumption of the SU, equivalently). From
\eqref{Eq_UTR_power}, $\bar{P}_c$ equals to \vspace{-2mm}
\begin{equation} \label{Eq_sum_power}
\bar{P}_c=t_2P_2+t_3P^{(3)}_2+ t_3P_3=\bar{P}_2+\bar{P}_3.
\end{equation}
Now we have the following Corollary with the proof given in
Appendix \ref{App_mul_UTR}.

\begin{coro} \label{Theorem_mul_UTR}
With full CSIT, the following multiplexing gain of the SU is
achievable by the CTR under the power constraints
\eqref{Eq_UTR_power} \vspace{-3mm}
\begin{equation}\label{eq_mul_UTR}
\max \left\{m,\left(1-\frac{C(|h_{14}|^2\stackrel
{\;}{\bar{P}_1})}{C(|h_{12}|^2\stackrel {\;}{\bar{P}_1})}\right)^+
\right\},
\end{equation}
where \vspace{-6mm}
\begin{equation} \label{Eq_UTR_MG}
m=\Bigg\{\begin{array}{ll}1-t_3, \mbox{ for any $t_3$}\in
\left(0,\,1-\left(C(|h_{14}|^2\stackrel
{\;}{\bar{P}_1})/C(|h_{13}|^2\stackrel
{\;}{\bar{P}_1})\right)\right], &\mbox{ when } |h_{14}|<|h_{13}|,
\vspace{-3mm}
\\ 0,& \mbox{    otherwise.}\end{array}
\end{equation}
\vspace{-5mm}
\end{coro}

Indeed, according to Appendix \ref{App_mul_UTR}, the multiplexing
gains $m$ and $(1-C(|h_{14}|^2\stackrel
{\;}{\bar{P}_1})/C(|h_{12}|^2\stackrel {\;}{\bar{P}_1}))^+$
correspond to the CTR using pure receiver ($t_1=0$) and pure
transmitter ($t_3=0$) relaying, respectively. Comparing
\eqref{eq_mul_UTR} and \eqref{Eq_mul_CT}, we know that with full
CSIT, the multiplexing gain of the CTR is larger (or at least no
less) than that of the CT (also its special case in
\cite{Viswanath_CR}). In the next section, we will investigate the
performance of the CTR and CT in fast Rayleigh faded channels with
only the statistics of CSIT. The CTR is even more promising in
this setting. \vspace{-3mm}
\section{Performance in Fast Rayleigh Fading Channels with Statistics of CSIT} \label{sec_ergodic}
\vspace{-1mm} We will first show that the performance of CT (and
its special case \cite{Viswanath_CR}) has rate performance worse
than that of the CTR. Then we focus on the CTR and its achievable
rate. The optimal common message relaying ratio $\alpha_1$ will
also be investigated. First, for the precoding for the CT, it was
shown that the linear-assignment Gel'fand-Pinsker coding (LA-GPC)
\cite{DPC_GelPin} outperforms the DPC in Ricean-faded cognitive
channels with the statistics of CSIT \cite{scCR_TWCOMs09}. This is
because the LA-GPC, which includes the DPC as a special case, does
not need the full CSIT as the DPC in designing the precoding
paramteters. However, for Rayleigh fading channels with only the
statistics of CSIT, we observe that even the more general LA-GPC
results in a rate performance the same as that of treating
interference as noise. So the CTR will outperform the DPC based CT
in this channel setting. With a little abuse of notations, the
above observation can be found as the following proposition with
the proof given in Appendix \ref{App_bad_DPC}.

\begin{proposition} \label{Propo_bad_DPC}
With only the statistics of CSIT, for the ergodic Rayleigh faded
channel $Y_3=H_{23}X_2+H_{13}X_1+Z_3$ with transmitter
side-information $X_1$ and power constraints $\E[|X_1|^2] \leq
\bar{P}_1,\E[|X_2|^2] \leq \bar{P}_2$, the maximal achievable rate
of the LA-GPC coded $X_2$ is the same as the rate obtained by
treating the interference $H_{13}X_1$ as noise, which is
\vspace{-5mm}
\begin{equation} \label{Eq_LA_GPC}
\E\left[C
\left(\frac{|H_{23}|^2\bar{P}_2}{1+|H_{13}|^2\stackrel
{\;}{\bar{P}_1}}\right)\right].
\end{equation}
\end{proposition}
\vspace{1mm}

It is easy to use Proposition \ref{Propo_bad_DPC} to calculate the
achievable rate of CT, which equals to the rate of treating
$H_{13}X_1$ at Node 3 in Phase 2 as noise. Then the CTR always
performs better than the CT in the fast Rayleigh fading channels
according to the following intuitions. In the CTR, Node 3 will
face a two user MAC in Phase 2, and the rate pair from treating
$H_{13}X_1$ as noise while decoding the SU's message is always in
the rate region of this MAC. Thus we only describe the CTR and its
achievable rate in detail as follows
\\
\noindent \textbf{Phase 1}: In the first $\lfloor t_1n \rfloor $
symbols, Node 2 and 3 listen to the PU's message $w_1$. Node 2
decodes $w_1$.

\noindent \textbf{Phase 2}: Within the next $\lfloor t_2n \rfloor$
symbols, Node 2 transmits
\begin{equation} \label{Eq_UTR_X2ergo}
    X_2=U_2+\sqrt{\alpha_1\frac{P_2}{P_1}}X_1,
\end{equation}
where $\alpha_1$ is the relaying ratio for the common message
$w_1$. Node 3 listens to and decodes $w_1$ and $w_2$.

\noindent \textbf{Phase 3}: For the rest of $\lfloor t_3n \rfloor$
symbol time, the clean relaying signals are transmitted from Node
2 and 3 respectively as \vspace{-3mm}
\begin{equation} \label{Eq_UTR_X2_Phase3ergo}
X_2= \sqrt{P^{(3)}_2/P_1}X_1 \mbox{ and } \; X_3= \sqrt{P_3/P_1}X_1.
\end{equation}
\noindent Note that one of the differences compared with the full
CSIT case in Section \ref{sec_UTR} is that now the CTR cannot
chose the phase in \eqref{Eq_UTR_X2ergo} and
\eqref{Eq_UTR_X2_Phase3ergo} since the channel phase realizations
are unknown at Node 2.

The achievable rate of the CTR in fading channels is presented in
the following Theorem. Compared with the conventional fast fading
channels, now the channel fading statistics will vary in different
phases (block of symbols) at Node 3 and 4. This new problem
corresponds to the SNR variation problem in Section \ref{sec_UTR},
and can be solved by Lemma \ref{Lemma_WeiYu} and
\ref{Lemma_WeiYuMAC} as well as the channel ergodicity. The
detailed proof is given in Appendix \ref{App_UTR_ergo}.

\begin{theorem} \label{Theorem_UTR_ergo}
With the statistics of CSIT and the transmitted power meeting
\eqref{Eq_UTR_power}, the following rate of the SU is achievable
by the CTR in the fast Rayleigh faded channel
\begin{equation} \label{eq_UTR_CR_ergo}
\!R_2\!\leq \max_{t_2,\alpha_1}
\min\!\Bigg\{\!t_2\E\!\!\left[C\!\left(|H_{23}|^2(1-\alpha_1)P_2\right)\right]\!,t_2\E\!\!\left[\!C\!\left(\left|H_{13}\!+\!\!\!\sqrt{\frac{\alpha_1
P_2}{P_1}}H_{23}\right|^2\!\!P_1+|H_{23}|^2(1-\alpha_1)P_2\right)\!\right]\!\!-R_T^m\!\Bigg\}\!,\!\!
\end{equation}
where $R^m_T= \min\left\{R'_T,\,t_2 \cdot \bE\left[C
\left(\big|H_{23}\sqrt{\alpha_1P_2/P_1}+H_{13}\big|^2P_1\right)\right]\right\}$
with $R'_T \triangleq R_T-t_1\E[C(|H_{13}|^2P_1)]$, and is subject
to the constraint for coexistence with \vspace{-3mm}
\begin{align} \label{eq_UTR_coex_ergo}
R_T \leq &t_1\E[C(|H_{14}|^2P_1)]+
 t_2\E\left[C\left(\frac{\left|H_{14}+\sqrt{\frac{\alpha_1
P_2}{P_1}}H_{24}\right|^2P_1}{1+|H_{24}|^2(1-\alpha_1)P_2}\right)\right]\notag\\
&+(1-t_1-t_2)\E\left[C\left(\left|H_{14}+\sqrt{P^{(3)}_2/P_1}H_{24}+\sqrt{P_3/P_1}H_{34}\right|^2P_1\right)\right],
\end{align}
where the $\alpha_1$, $t_1$ and $t_2$ are defined as those in
Theorem \ref{Theorem_UTR}, respectively. Moreover, let $u_{T \!\!
x}=\mathbf{1}_{t_1>R_T/\bE[C(|H_{12}|^2P_1)]}$ and
$u_{Rx}=\mathbf{1}_{t_2\bE[C(|H_{23}\sqrt{\alpha_1P_2/P_1}+H_{13}|^2P_1)]\geq
R_T'}$, $t_1$, $\alpha_1$, and $P^{(3)}_2$ are all zero if $u_{T
\!\! x}=0$, while $P_3=0$ if $u_{Rx}=0$.
\end{theorem}
\vspace{+1mm}

Unlike the full CSIT case, we can characterize the optimal common
message relaying ratio $\alpha_1$ as in the following Corollary.
The key observation is that the pointwise minimum of the two rate
functions in \eqref{eq_UTR_CR_ergo} can be shown to be
monotonically decreasing with $\alpha_1$. Note that we can not get
similar results for the full CSIT case, the discussions are given
right after the proof of this Corollary.

\begin{coro} \label{Coro_ergodic}
Given $t_1$ and $t_2$, the optimal common message relaying ratio
$\alpha_1$ in Theorem \ref{Theorem_UTR_ergo} will validate the
equality in the constraint for coexistence
\eqref{eq_UTR_coex_ergo}.
\end{coro}
\begin{proof}
To get the desire result, first we prove that both arguments of
the pointwise minimum $\min\{,\}$ in \eqref{eq_UTR_CR_ergo} are
monotonically decreasing with $\alpha_1$ given $t_2$. We focus on
the second argument first and rearrange it as \vspace{-4mm}
\begin{equation} \label{eq_UTR_ero_alpha}
t_2
\E\bigg[\log\bigg(1+|H_{13}|^2P_1+|H_{23}|^2P_2+2\mathrm{Re}\{H_{13}H^*_{23}\}\sqrt{P_1P_2}\sqrt{\alpha_1}\bigg)\bigg]-R^c_T
=t_2 \max\{f_1(\alpha),f_2(\alpha)\},
\end{equation}
where the equality comes from the definition of $R_T^m$ in Theorem
\ref{Theorem_UTR_ergo}, with $f_1(\alpha)$ and $f_2(\alpha)$
defined as
\begin{align}
f_1(\alpha) \triangleq
&\E\bigg[\log\bigg(1+|H_{13}|^2P_1+|H_{23}|^2P_2+2\mathrm{Re}\{H_{13}H^*_{23}\}\sqrt{P_1P_2}\sqrt{\alpha_1}\bigg)\bigg]-
\frac{R'_T}{t_2},\label{eq_ergo_alpha_f1}\\
f_2(\alpha)\triangleq
&\E\bigg[\log\bigg(1+|H_{13}|^2P_1+|H_{23}|^2P_2+2\mathrm{Re}\{H_{13}H^*_{23}\}\sqrt{P_1P_2}\sqrt{\alpha_1}\bigg)\bigg]
\notag \\ & - \bE\left[C
\left(\big|H_{23}\sqrt{\alpha_1P_2/P_1}+H_{13}\big|^2P_1\right)\right],
\label{eq_ergo_alpha_f2}
\end{align}
respectively. In the following, we will respectively show that
$f_1(\alpha)$ and $f_2(\alpha)$ are both monotonically decreasing
of $\alpha_1$. Since the pointwise maximum of the two
monotonically decreasing functions is still a monotonically
decreasing function, from \eqref{eq_UTR_ero_alpha}, the second
argument of the $\min\{,\}$ in \eqref{eq_UTR_CR_ergo} is a
monotonically decreasing function of $\alpha_1$.

Now we show the monotonically decreasing properties of
$f_1(\alpha)$ and $f_2(\alpha)$. As for the $f_1(\alpha)$ in
\eqref{eq_ergo_alpha_f1}, note that from the definition of $R'_T$
in Theorem \ref{Theorem_UTR_ergo}, only the first term in the RHS
of \eqref{eq_ergo_alpha_f1} is related to $\alpha_1$. This term
can be further represented by
\begin{equation} \label{eq_UTR_ero_alpha1}
\E_{|H_{13}|,|H_{23}|}\bigg[\E_{\theta_{13},\theta_{23}}\bigg[\log\bigg(1+|H_{13}|^2P_1+|H_{23}|^2P_2+2\sqrt{P_1P_2}\sqrt{\alpha_1}\mathrm{Re}\{H_{13}H^*_{23}\}\bigg)\bigg||H_{13}|,|H_{23}|\bigg]\bigg],
\end{equation}
where the property of the conditional mean is applied. We will
show that given realizations $|H_{13}|=|h_{13}|$ and
$|H_{23}|=|h_{23}|$, the conditional mean
$\E_{\theta_{13},\theta_{23}}\big[(.)\big|\scriptstyle{|H_{13}|=|h_{13}|,|H_{23}|=|h_{23}|}\big]$
in \eqref{eq_UTR_ero_alpha1} is a monotonically decreasing
function of $\alpha_1$. Then so are \eqref{eq_UTR_ero_alpha1} and
$f_1(\alpha)$. This conditional mean equals to
\begin{equation} \label{eq_UTR_ero_alpha2}
\E_{\theta_{13},\theta_{23}}\left[\log\bigg(1+|h_{13}|^2P_1+|h_{23}|^2P_2+2\sqrt{P_1P_2}\sqrt{\alpha_1}|h_{13}||h_{23}|\cos(\theta_{13}-\theta_{23})\bigg)\right].
\end{equation}
Since $|H_{13}|,|H_{23}|$ and $\theta_{13},\theta_{23}$ are
independent, given $|H_{13}|=|h_{13}| \mbox{ and
}|H_{23}|=|h_{23}|$, both $\theta_{13}$ and $\theta_{23}$ are
still independent and uniformly distributed in $(0,2\pi]$,
respectively. Then $\cos(\theta_{13}-\theta_{23})$ is zero mean.
Together with the fact that the log function is concave, we know
that \eqref{eq_UTR_ero_alpha2} is monotonically decreasing with
respect to $\alpha_1$ from \cite[P.115]{Book_Boyd}. As for
$f_2(\alpha)$, note that the term
$\bE[C(|h_{23}\sqrt{\alpha_1P_2/P_1}+h_{13}|^2P_1)]$ in
\eqref{eq_ergo_alpha_f2} is monotonically increasing in
$\alpha_1$. Since the first terms of the RHS of
\eqref{eq_ergo_alpha_f2} and \eqref{eq_ergo_alpha_f1} are the
same, from the previous results, we establish the monotonically
decreasing property of $f_2(\alpha)$.

As for
$\!t_2\E\!\!\left[C\!\left(|H_{23}|^2(1-\alpha_1)P_2\right)\right]$,
the first argument of the $\min\{,\}$ in \eqref{eq_UTR_CR_ergo},
it is clear that this term is monotonically decreasing with
$\alpha_1$ given $t_2$. Then from the fact that the minimum of two
monotonically decreasing functions results in a monotonically
decreasing function, we prove the monotonically decreasing
property of the pointwise minimum in \eqref{eq_UTR_CR_ergo}.
Finally, it is easy to see that the RHS of
\eqref{eq_UTR_coex_ergo} monotonically increases with $\alpha_1$
given $t_1$ and $t_2$. Then the optimal $\alpha_1$ must validates
the equality in \eqref{eq_UTR_coex_ergo}.
\end{proof}

Note that the optimization problem with full CSIT in Theorem
\ref{Theorem_UTR} is much more complicated than that in Theorem
\ref{Theorem_UTR_ergo}, and the simple result in Corollary
\ref{Coro_ergodic} can not be obtained. Depending on the
combinations of $\theta_{13}$,$\theta_{23}$ and $\theta_1$, the
second argument of the $\min\{,\}$ in \eqref{eq_UTR_CR} may
increase with $\alpha_1$. That is, more common message relaying
from Node 2 can increase the sum rate of the MAC at Node 3. The
monotonically decreasing property does not always exist in the RHS
of \eqref{eq_UTR_CR}, and the SU's rate may increases in a certain
range of $\alpha_1$. However, the unknown channel phase at Node 2
prohibits the SU to adjust $\theta_1$, and the common message
relaying is blind and always harmful at Node 3. One should just
use the minimum power which meets the constraint for coexistence
for the common message relaying.

Now we show the multiplexing gain. The proof is similar to that of
Corollary \ref{Theorem_mul_UTR} and is omitted.

\begin{coro} \label{Theorem_mul_UTR_ergo}
With the statistics of CSIT, the CTR can achieve the following
multiplexing gain under the power constraints
\eqref{Eq_UTR_power}, \vspace{-3mm}
\begin{equation}\label{eq_mul_UTR_fading}
\max \left\{m,\left(1-\frac{\bE[C(|H_{14}|^2\stackrel
{\;}{\bar{P}_1})]}{\bE[C(|H_{12}|^2\stackrel
{\;}{\bar{P}_1})]}\right)^+ \right\},
\end{equation}
where \vspace{-6mm}
\[
m=\Bigg\{\begin{array}{ll}1-t_3, \mbox{ for any $t_3$}\in \left(0,\,1-\frac{\bE[C(|H_{14}|^2\stackrel {\;}{\bar{P}_1})]}{\bE[C(|H_{13}|^2\stackrel {\;}{\bar{P}_1})]}\right], & \mbox{ when }\bE[C( |H_{14}|^2\bar{P}_1)]<\bE[C(|H_{13}|^2\bar{P}_1)], \vspace{-4mm}\\
0, &\mbox{    otherwise.}\end{array}
\]
\end{coro}
\vspace{-5mm}
\section{Simulation Results}\label{Sec_Simu}
\vspace{-2mm} Here we provide simulation results to show the
performances of our clean-relaying aided cognitive radios. In the
following discussions and the simulation figures, we will
abbreviate the results from \cite{Viswanath_CR}, or CT with
$t_3=0$, as JV. The noise variances at the receivers are set to
unity, and the average transmitted SNR of PU ($\bar{P}_1$ in
\eqref{Eq_avg_power}) is set to 20 dB. We assume that the SU in
both CT (including JV) and CTR have the same average transmission
SNR $\bar{P}_c$, which can be computed according to
\eqref{Eq_UT_power} ($\bar{P}_c=\bar{P}_2$) and
\eqref{Eq_sum_power}, respectively. We set the PU's rate $R_T$ as
that when the interference from the SU is absent, that is, as
$C(|h_{14}|^2 P_1)$ and $\E[C(|H_{14}|^2 P_1)]$ in the full and
statistics of CSIT cases, respectively.

We first show the rate comparisons for channels with full CSIT.
The channel gain of each figure is listed in Table
\ref{Table_Channel_Full} where the unit of the phase is radian.
The $t_1$ in both CT and JV are $R_T/C(|h_{12}|^2P_1).$ In
Fig.~\ref{Fig_rate_SNR}, we can see that with large enough
$|h_{34}|$ as specified in Table \ref{Table_Channel_Full}, the
clean relaying from Node 3 makes the CTR have the best rate
performance. Next we consider the case where $|h_{34}|$ is weaker
in Fig.~\ref{Fig_UT_good}. When $|h_{34}|$ is smaller than
$|h_{24}|$, the CTR may prefer clean relaying from Node 2 rather
than from Node 3, that is, $P^{(3)}_2>P_3=0$. It is easy to check
that in this case, the optimal $\alpha_1$ for the CTR is also
feasible for the CT. Then comparing \eqref{eq_UTR_CR} and
\eqref{eq_UT_CR}, we know that the CT performs better than the CTR
as in Fig.~\ref{Fig_UT_good}. Moreover, in Fig.~\ref{Fig_rate_SNR}
and \ref{Fig_UT_good}, the clean relaying of the CTR and CT yields
significant gains over the JV, respectively. Next, we show how the
SU's rate changes with $|h_{24}|$ in Fig.~\ref{Fig_rate_h24}. We
can find out that there are three regions. In Region 1, where
$|h_{24}|<|h_{14}|$, we find that the CT and JV coincide. This is
consistent with \cite{Viswanath_CR}, where JV is proved to be
optimal in this region when relaying from Node 3 is prohibited. In
Region 2 and 3, $|h_{24}|>|h_{14}|$, the JV wastes lots of power
on the relaying since the SU produces large interference at Node
4. The CT performs better than the JV due to the clean relaying.
In Region 2, $|h_{24}|<|h_{34}|$, the CTR performs better than the
CT since the CTR can use a better relaying path than that of CT in
Phase 3. In Region 3, $|h_{24}|>|h_{34}|$, the CT performs the
best according to previously discussions for
Fig.~\ref{Fig_UT_good}. However, the CT and CTR have the same
performance due to the following reasons. In the channel setting
for Fig.~\ref{Fig_rate_h24} listed in Table
\ref{Table_Channel_Full}, we find that the first term of the
$\min\{.\}$ in \eqref{eq_UTR_CR} is selected, which is the same as
\eqref{eq_UT_CR}. Moreover, since this term is independent of
$\theta_1$, the relaying phase $\theta_1$ for the CTR is chosen as
$\theta_{14}-\theta_{24}$ from \eqref{eq_UTR_coex}. Together with
the power allocation as in the discussions for
Fig.~\ref{Fig_UT_good}, the constraints for coexistence
\eqref{eq_UTR_coex} and \eqref{eq_UT_coex} are the same in this
simulation. The optimal $\alpha_1$ of CTR and CT are also the
same, and the CTR and CT have the same rate performance.

Next we consider the rate performance in the fast Rayleigh faded
channels with the statistics of CSIT. The channel variance of each
link is listed in Table \ref{Table_Channel_fading}. As shown in
Fig. \ref{Fig_ergodic_h24_leq_h34}, the CTR outperforms the CT and
JV, which is consistent with the discussions under Proposition
\ref{Propo_bad_DPC} in Section \ref{sec_ergodic}. The $t_1$ in the
JV is set to $R_T/\E[C(|h_{12}|^2P_1)]$. When the SU's transmitted
SNR is low, the CT (also JV) can only support very low rate as
shown in Fig.~\ref{Fig_ergodic_h24_leq_h34}. This is because that
the PU's transmitted SNR is set to 20 dB, then the interference at
Node 3 is relatively large for the SU when the SU's transmitted
SNR is small. According to Proposition \ref{Propo_bad_DPC}, the SU
of CT (also JV) can only treat interference from the PU as noise,
which degrades the rate performance a lot. However, the MAC
decoder of CTR at Node 3 can avoid this problem. In Fig.
\ref{Fig_coro2} we show an example to verify the results in
Corollary \ref{Coro_ergodic}. We can find that the optimal
$\alpha_1$ which maximizes the SU's rate also make the equality in
the constraint for coexistence \eqref{eq_UTR_coex_ergo} valid.
That is, the optimal $\alpha_1$ is the minimum $\alpha_1$ which
makes the PU's rate with the interference from SU the same as the
interference-free rate.

Finally, we show the multiplexing gain comparisons in the
following. Following the spirit of \cite{PaulrajJSAC07}, we use
the generalized multiplexing gain (GMG) of the SU, which is
defined as $R_2/\log\bar{P}_c$, as the performance metric for
finite SNR. As $\bar{P}_c$ approaches infinity, the GMG will
approach the multiplexing gain defined in \eqref{Eq_mul_def}. We
first show the full CSIT cases in Fig.~\ref{Fig_nonfading_MG} and
\ref{Fig_MG2} with channels specified in Table
\ref{Table_Channel_Full} respectively. In our simulation, we set a
lower bound for $t_3$ as 0.01 when $t_3 \neq 0$, and $m$ in
Corollary \ref{Theorem_mul_UTR} will be upper-bounded by
1-0.01=0.99. We then use the multiplexing gain in
\eqref{eq_mul_UTR} with $m=0.99$ as the GMG upper bound in
Fig.~\ref{Fig_nonfading_MG} and \ref{Fig_MG2}. With large
$|h_{34}|$ as in Table \ref{Table_Channel_Full}, the GMG
advantages of the CTR over the JV can be seen from
Fig.~\ref{Fig_nonfading_MG}. When the transmitted SNR is larger
than 40 dB, we can find that the curve of CTR diverges from those
of the CT and JV. This is because the CTR selects pure receiver
relaying in this SNR region. Since $|h_{14}|<|h_{13}|$ in this
simulation, according to discussions under Corollary
\ref{Theorem_mul_UTR}, the CTR with pure receiver relaying has
larger GMG than those of the CT and JV when $t_3$ is small and the
SNR is large enough. Also when the SNR increases, the GMG of the
CTR will approach the upper bound \eqref{eq_mul_UTR}. Note that we
plot the figures according to the transmitted SNR not the common
received SNR in most of the literatures. The transmitted SNR is
much larger than the received SNR since the $|h_{23}|$ of
Fig.~\ref{Fig_nonfading_MG} in Table \ref{Table_Channel_Full} is
small. It then takes larger transmit SNR than the common received
SNR for the GMG to approach the upper bound (multiplexing gain).
In Fig. \ref{Fig_MG2}, we show the case with small $|h_{34}|$. The
CT performs the best while the CTR performances the worst.
However, as predicted by Corollary \ref{Theorem_mul_UTR}, even
though the CTR has the worst GMG, it will approach the GMGs of CT
and JV as the SNR increases. The GMG results for the fading
channels with the statistics of CSIT are shown in
Fig.~\ref{Fig_fading_MG1}. The GMG upper bound is computed from
Corollary \ref{Theorem_mul_UTR_ergo} with $m=0.99$ as in
Fig.~\ref{Fig_nonfading_MG}. According to the discussions for Fig.
\ref{Fig_ergodic_h24_leq_h34}, the CT and JV always have worse GMG
than that of the CTR according to Proposition \ref{Propo_bad_DPC}.

\vspace{-4mm}
\section{Conclusion}\label{Sec_Conclusion}
\vspace{-2mm} In this paper, we considered the
interference-mitigation based cognitive radio where the SU must
meet the coexistence constraint to maintain the rate performance
of the PU. We proposed two new transmission schemes aided by the
clean relaying named as the clean transmitter relaying and the
clean transmitter-receiver relaying aided cognitive radio,
respectively. Compared with the previous DPC-based cognitive radio
without clean relaying, the proposed schemes provide significant
rate gains in a variety of channels with different levels of CSIT.
Moreover, the implementation complexity of the CTR is much lower
than that of the DPC-based cognitive radio. \vspace{-5mm}
\appendix 
\vspace{-4mm}
\subsection{Proof of Theorem \ref{Theorem_UT}} \label{App_CT}
\vspace{-2mm} Let $X_1$ be zero mean Gaussian with variance $P_1$,
the PU then generates its random codebook according to the
distribution of $X_1$ with rate $R_T$. From
\cite{Azarian_advance_decoding} we know that the fractional
decoding interval must satisfy $t_1>R_T/I(X_1;Y_2)$ to ensure the
successful decoding of $w_1$ using the received symbols from Node
2 in Phase 1. It then results in the constraint \eqref{eq_UT_DF}
from \eqref{Eq_channel}.

We now invoke Lemma \ref{Lemma_WeiYu} to derive the coexistence
constraint. From \eqref{Eq_Vish_X2} in Phase 2,
\eqref{Eq_UT_X2_Phase2} in Phase 3 and the channel model
\eqref{Eq_channel}, we know that within the $n$-symbol time,
$K_{X_1^n}=P_1\mathbf{I}_{n}$, the equivalent channel at Node 4
\vspace{-4mm}
\begin{align}
\bH^n=\mathrm{diag}\Bigg(& h_{14}\mathbf{I}_{\lfloor t_1n
\rfloor}, \;
\left(|h_{14}|+|h_{24}|\sqrt{\alpha_1\frac{P_2}{P_1}}\right)e^{j\theta_{14}}\mathbf{I}_{\lfloor
t_2n \rfloor}, \notag
\left(|h_{14}|+|h_{24}|\sqrt{P_2/P_1}\right)\e^{j\theta_{14}}\mathbf{I}_{\lfloor
t_3n \rfloor} \Bigg) \label{UT_coex_H}
\end{align}
and the equivalent noise has covariance matrix $
K_{Z^n}=\mathrm{diag}\left( \mathbf{I}_{\lfloor t_1n \rfloor},
(1+|h_{24}|^2(1-\alpha_1)P_2)\mathbf{I}_{\lfloor t_2n
\rfloor},\bI_{\lfloor t_3n \rfloor} \right), $ since the DPC
encoded $X^D_2$ is Gaussian with variance $(1-\alpha_1)P_2$ and
independent of $X_1$ \cite{CostaDPC}. Then by invoking Lemma
\ref{Lemma_WeiYu}, we have \eqref{eq_UT_coex} to ensure that $R_T$
is single-user achievable. Finally, since Node 2 uses $h_{13}X_1$
as the noncausal side-information at the transmitter in Phase 2,
by applying the well-known DPC result \cite{CostaDPC} we have
\eqref{eq_UT_CR}.

\vspace{-7mm}
\subsection{Proof of Corollary \ref{Theorem_mul_UTR}} \label{App_mul_UTR}
\vspace{-2mm} We will consider two cases, that is, pure receiver
and pure transmitter relaying. These two schemes can achieve
multiplexing gains $m$ and $(1-C(|h_{14}|^2\stackrel
{\;}{\bar{P}_1})/C(|h_{12}|^2\stackrel {\;}{\bar{P}_1}))^+$,
respectively. When the channels conditions $|h_{13}|>|h_{14}|$ and
$|h_{12}|>|h_{14}|$ are both valid, both schemes are feasible and
the CTR can achievable the best multiplexing gain of these two
schemes as \eqref{eq_mul_UTR}. If only one of the channel
conditions is valid, the multiplexing gain of the corresponding
feasible scheme will be chosen by \eqref{eq_mul_UTR}.

We first show that if $|h_{13}|>|h_{14}|$, as \eqref{Eq_UTR_MG},
the multiplexing gain $1-t_3$ is achievable by the pure receiver
relaying. In this scheme, $t_1=\alpha_1=0$, $P_2^{(3)}=0$, and
$t_3=1-t_2$, then we may set $P_3=P_2=\bar{P}_c$ from
\eqref{Eq_sum_power}. Without loss of generality, we can set $R_T
= C(|h_{14}|^2\stackrel {\;}{\bar{P}_1})$ in the following
analysis since $R_T \leq C(|h_{14}|^2\stackrel {\;}{\bar{P}_1})$
from the channel capacity theorem \cite{Book_Cover}. With the
above parameter selections, the constraint for coexistence
\eqref{eq_UTR_coex}, and the constraint
$t_2C(|h_{23}\sqrt{\alpha_1P_2/P_1}\e^{j
\theta_1}+h_{13}|^2P_1)\geq R_T'$ to validate $u_{Rx}=1$
respectively reduce to
\begin{equation} \label{eq_UTR_coex_mul}
C(|h_{14}|^2\stackrel {\;}{\bar{P}_1})<t_2 \cdot
C\!\left(\!\frac{|h_{14}|^2\bar{P}_1}{1+|h_{24}|^2\stackrel
{\;}{\bar{P}_c}}\!\right)+(1-t_2)C\left(\!\left(|h_{14}|\sqrt{\stackrel
{\;}{\bar{P}_1}}+|h_{34}|\sqrt{\stackrel
{\;}{\bar{P}_c}}\right)^{\!2}\!\right)\!, \mbox{and}\;\; t_2 \geq
\frac{C(|h_{14}|^2\stackrel
{\;}{\bar{P}_1})}{C(|h_{13}|^2\stackrel {\;}{\bar{P}_1})}.
\end{equation}
When $\bar{P}_c \rightarrow \infty$, we can find that the range of
$t_2$ to validate (\ref{eq_UTR_coex_mul}) is $
\frac{C(|h_{14}|^2\stackrel
{\;}{\bar{P}_1})}{C(|h_{13}|^2\stackrel {\;}{\bar{P}_1})}\leq t_2
<1. $ Therefore we need $t_3\in
(0,\,1-C(|h_{14}|^2\bar{P}_1)/C(|h_{13}|^2\bar{P}_1)]$ to meet the
constraints. From \eqref{eq_UTR_CR}, \eqref{Eq_mul_def} and the
fact that $R^m_T=R'_T$ since $u_{Rx}=1$, it is easy to see that
the multiplexing gain $t_2=1-t_3$ is achievable, and
\eqref{Eq_UTR_MG} is valid. Note that our selection of $t_2$ and
$\alpha_1$ is definitely a suboptimal choice with respect to
\eqref{eq_UTR_CR}. If $|h_{13}|\leq |h_{14}|$ and $t_1=0$, there
will be no relaying in this case since Node 3 can not decode $w_1$
before the end of Phase 2. Then the multiplexing gain is zero for
pure receiver relaying as in \eqref{Eq_UTR_MG}.

Now we show that when $|h_{14}|<|h_{12}|$, the multiplexing gain
$1-C(|h_{14}|^2\stackrel {\;}{\bar{P}_1})/C(|h_{12}|^2\stackrel
{\;}{\bar{P}_1})$ in \eqref{eq_mul_UTR} is achievable with only
transmitter relaying ($t_3=0$). To prove this, we sub-optimally
set $t_1=\frac{C(|h_{14}|^2\stackrel
{\;}{\bar{P}_1})}{C(|h_{12}|^2\stackrel {\;}{\bar{P}_1})}$,
$t_2=1-t_1$ and $\theta_1=\theta_{14}-\theta_{24}$. Together with
the setting $R_T = C(|h_{14}|^2\stackrel {\;}{\bar{P}_1})$ as
describe previously, the coexistence constraint in
\eqref{eq_UTR_coex} then becomes \vspace{-3mm}
\begin{equation} \label{eq_mul_UTR_OT_coexs}
C(|h_{14}|^2{P}_1)<C\Bigg(\frac{\big(|h_{14}|+|h_{24}|\sqrt{\alpha_1P_2/{P}_1}\big)^2{P}_1}{1+|h_{24}|^2(1-\alpha_1)P_2}\Bigg).
\end{equation}
With $t_2P_2=\bar{P}_c$ from \eqref{Eq_sum_power} and
$P_1=\stackrel {\;}{\bar{P}_1}$ from \eqref{Eq_UTR_power}, as
$\bar{P}_c \rightarrow \infty$, \eqref{eq_mul_UTR_OT_coexs}
becomes $ |h_{14}|^2\stackrel
{\;}{\bar{P}_1}<\frac{\alpha_1}{1-\alpha_1}. $ Then we have
$\alpha_1>|h_{14}|^2\stackrel
{\;}{\bar{P}_1}/(1+|h_{14}|^2\stackrel {\;}{\bar{P}_1})$ to meet
the constraint for coexistence. It can be easily seen that with
the selected $\alpha_1$, $\theta_1$, and $t_2$, when
$\bar{P}_c\rightarrow \infty$, $R^m_T=R'_T$ in \eqref{eq_UTR_CR}.
Therefore, from \eqref{eq_UTR_CR} and \eqref{Eq_mul_def} we can
find that the multiplexing gain $t_2=1-C(|h_{14}|^2\stackrel
{\;}{\bar{P}_1})/C(|h_{12}|^2\stackrel {\;}{\bar{P}_1})$ is
achievable. Finally, when $|h_{14}| \geq |h_{12}|$ the function
$(.)^+$ in \eqref{eq_mul_UTR} will force the multiplexing gain to
be zero. In this case, the coexistence constraint is violated
since Node 2 cannot relay without correct knowledge of $w_1$.
\vspace{-7mm}
\subsection{Proof of Proposition \ref{Propo_bad_DPC}} \label{App_bad_DPC}
\vspace{-2mm} From \cite{scCR_TWCOMs09}, by treating $X_1$ as
non-causally known transmitter side-information, the following
rate is achievable by the LA-GPC \vspace{-5mm}
\begin{equation} \label{Eq_LA_GPC_r}
\max_{\beta } \{\E\big[\log
\big((|H_{23}|^2\bar{P}_2+|H_{13}|^2\bar{P}_1+1)\bar{P}_2\big)\big]-f(\beta)\},
\end{equation}
where $ f(\beta) \triangleq \E\big[\log
\big(\bar{P}_1\bar{P}_2|H_{13}-\beta
H_{23}|^2+\bar{P}_2+|\beta|^2\bar{P}_1\big)\big], $ and $\beta\in
\mathds{C}$ is the precoding coefficient of the LA-GPC. Note that
solving \eqref{Eq_LA_GPC_r} over $\beta$ is the same as minimizing
$f(\beta)$. In the following we will show that $f(0)$ is the
minimal. We know that for any $\beta$ \vspace{-2mm}
\begin{align}  \notag
f(0)=\E\big[\log \big(\bar{P}_1\bar{P}_2|H_{13}|^2 +\bar{P}_2
\big) \big]&\leq
\E\big[\log \big(\bar{P}_1\bar{P}_2(1+|\beta|^2\sigma^2_{23}/\sigma^2_{13})|H_{13}|^2+\bar{P}_2 \big) \big] \\
&= \E\big[\log \big(\bar{P}_1\bar{P}_2|H_{13}-\beta
H_{23}|^2+\bar{P}_2 \big) \big], \label{Eq_LA_GPC_r1}
\end{align}
where the last equality comes from the fact that since $H_{23}$
and $H_{13}$ are independent zero-mean Gaussian distributed with
variance $\sigma^2_{23}$ and $\sigma^2_{13}$, respectively,
$H_{13}-\beta H_{23}$ is also zero-mean Gaussian distributed with
variance $\sigma^2_{13}+|\beta|^2\sigma^2_{23}$. Thus
$(1+|\beta|^2\sigma^2_{23}/\sigma^2_{13})|H_{13}|^2$ and
$|H_{13}-\beta H_{23}|^2$ have the same distribution. Moreover,
for any $\beta$, \vspace{-2mm}
\[
\E\big[\log \big(\bar{P}_1\bar{P}_2|H_{13}-\beta
H_{23}|^2+\bar{P}_2 \big) \big] \leq \E\big[\log
\big(\bar{P}_1\bar{P}_2|H_{13}-\beta
H_{23}|^2+\bar{P}_2+|\beta|^2\bar{P}_1\big)\big]=f(\beta).
\]
Combining the above equation with \eqref{Eq_LA_GPC_r1}, we know
that $\beta=0$ minimizes $f(\beta)$ and thus maximizes
\eqref{Eq_LA_GPC_r}. Substituting $\beta=0$ into
\eqref{Eq_LA_GPC_r} we get \eqref{Eq_LA_GPC}. \vspace{-7mm}
\subsection{Proof of Theorem \ref{Theorem_UTR_ergo}} \label{App_UTR_ergo}
\vspace{-2mm} To meet the coexistence constraint, we invoke Lemma
\ref{Lemma_WeiYu} again. Following the steps for proving
\eqref{eq_UTR_coex} in Theorem \ref{Theorem_UTR}, from
\eqref{Eq_UTR_X2ergo}, \eqref{Eq_UTR_X2_Phase3ergo},
\eqref{Eq_channel}, and Lemma \ref{Lemma_WeiYu}, to ensure that
the target PU's rate is single-user achievable
\begin{align}
R_T \leq &\frac{1}{n}\sum_{t=1}^{\lfloor t_1n \rfloor} \log
\left(1+|h_{14}(t)|^2P_1\right)+ \frac{1}{n}\sum_{t=\lfloor t_1n
\rfloor+1}^{\lfloor t_1n \rfloor+\lfloor t_2n \rfloor} \log
\left(1+\frac{\Big|h_{14}(t)+\sqrt{\frac{\alpha_1
P_2}{P_1}}h_{24}(t)\Big|^2P_1}{1+|h_{24}(t)|^2(1-\alpha_1)P_2)}\right)
\notag \\&+\frac{1}{n}\sum_{t=n-\lfloor t_3n \rfloor+1}^{n} \log
\left(1+\Bigg|h_{14}(t)+\sqrt{\frac{P^{(3)}_2}{P_1}}h_{24}(t)+\sqrt{\frac{P_3}{P_1}}h_{34}(t)\Bigg|^2P_1\right),
\label{eq_ergodic_Pombra}
\end{align}
where $h_{ij}(t)$ is the realization of the random channel
$H_{ij}$ at time $t$. When $n$ is large enough, the first term of
the RHS of \eqref{eq_ergodic_Pombra} can be rewritten as
\begin{equation}\label{Eq_1st_Ergodic_constraint}
\frac{1}{n}\sum_{t=1}^{\lfloor t_1n \rfloor} \log
\left(1+|h_{14}(t)|^2P_1\right)=t_1\frac{1}{\lfloor t_1n
\rfloor}\sum_{t=1}^{\lfloor t_1n \rfloor} \log
\left(1+|h_{14}(t)|^2P_1\right)=t_1\E[\log(1+|H_{14}|^2P_1)],
\end{equation}
where the last equality comes from the assumption that the channel
coefficients are i.i.d. and applying the ergodicity property.
After applying the same steps to the rest two terms of the RHS of
\eqref{eq_ergodic_Pombra}, we have the constraint for coexistence
\eqref{eq_UTR_coex_ergo}.

The achievable rate of the SU in \eqref{eq_UTR_CR_ergo} can be
obtained similarly. As for the steps to obtain
\eqref{eq_ergodic_Pombra}, we still invoke Lemma
\ref{Lemma_WeiYuMAC} but modify the proof steps of Theorem
\ref{Theorem_UTR} with the channel coefficients replaced by
$h_{ij}(t)$. Then we invoke the channel ergodicity as the proof
steps in \eqref{Eq_1st_Ergodic_constraint} to reach
\eqref{eq_UTR_CR_ergo}. The details are omitted. \vspace{-5mm}
{\renewcommand{\baselinestretch}{1.25}
\bibliographystyle{IEEEtran}
\bibliography{IEEEabrv,CodeSNps,DPC_PAPR2,scpub}
\newpage

\begin{table*}[ht]
\begin {center}
 \caption{Channel gains in \eqref{Eq_channel} Used in the Simulations (Full CSIT)}
\begin{tabular}{ccccccc}
\toprule

 Figure &$h_{14}$ & $h_{24}$ & $h_{34}$ & $h_{13}$ & $h_{23}$&
 $h_{12}$\\

\hline
   \ref{Fig_rate_SNR}   & $0.36e^{1.6j}$ & $0.45e^{1.6j}$ & $0.96e^{-3.1j}$ &$0.96e^{-0.69j}$ &$
0.24e^{-1.89j}$ &$e^{-2.28j}$ \\
   \ref{Fig_UT_good}   & $0.22e^{-1.6j}$ &$0.92e^{0.45j}$ & $0.74e^{1.19j}$ & $0.25e^{-0.69j}$ & $0.32e^{-1.89j}$ & $e^{1.4j}$
   \\
\ref{Fig_rate_h24}    &
 $0.22e^{-0.26j}$ & varying $|h_{24}|$, $\theta_{24}=\frac{\pi}{4}$&$0.32e^{-2.16j}$& $0.52e^{-0.95j}$ & $0.19e^{0.22j}$ & $e^{0.96j}$
\\
\ref{Fig_nonfading_MG} & $0.36e^{-0.78j}$ & $0.95e^{1.95j}$ &
$2.86e^{2.09j}$ & $0.96e^{0.87j}$ & $0.24e^{1.84j}$ &
$e^{-0.965j}$ \\

\ref{Fig_MG2} & $0.22e^{-1.6j}$ & $0.92e^{0.45j}$ &$0.74e^{1.19j}$
& $0.15e^{-0.69j}$ & $0.62e^{-1.89j}$ & $e^{1.4j}$ \\

 \bottomrule
\end{tabular} \label{Table_Channel_Full}
\end {center}
\end{table*}

\begin{table*}[ht]
\begin {center}
 \caption{Channel Variances of Rayleigh Fading Channels in
 \eqref{Eq_channel} \protect\\ Used in the Simulations (Statistics of CSIT)}
\begin{tabular}{ccccccc}
\toprule

 Figure &$\sigma^2_{14}$ & $\sigma^2_{24}$ & $\sigma^2_{34}$ & $\sigma^2_{13}$ & $\sigma^2_{23}$&
 $\sigma^2_{12}$\\

\hline
   \ref{Fig_ergodic_h24_leq_h34}   & 0.4 & 0.21 & 0.91 &0.82 &
   0.88 &1 \\
\ref{Fig_coro2} & 0.4 & 0.89 & 0.2 & 0.95 & 0.88 & 1 \\
\ref{Fig_fading_MG1} & 0.22 & 0.12 & 0.87 & 0.92 & 0.96 & 1 \\

 \bottomrule
\end{tabular} \label{Table_Channel_fading}
\end {center}
\end{table*}

\newpage
\begin{figure}
\centering
\epsfig{file=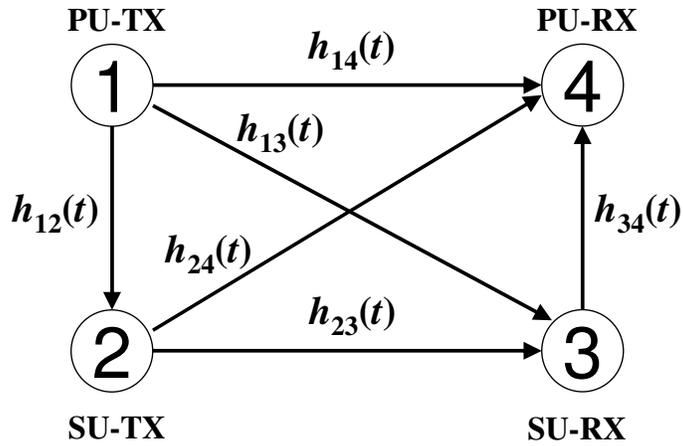,width=0.5\textwidth}
\caption{Cognitive channel model, where the TX and RX are the
abbreviations of the transmitter and receiver, respectively.}
\label{CR_channel}
\end{figure}

\begin{figure}
\centering \epsfig{file=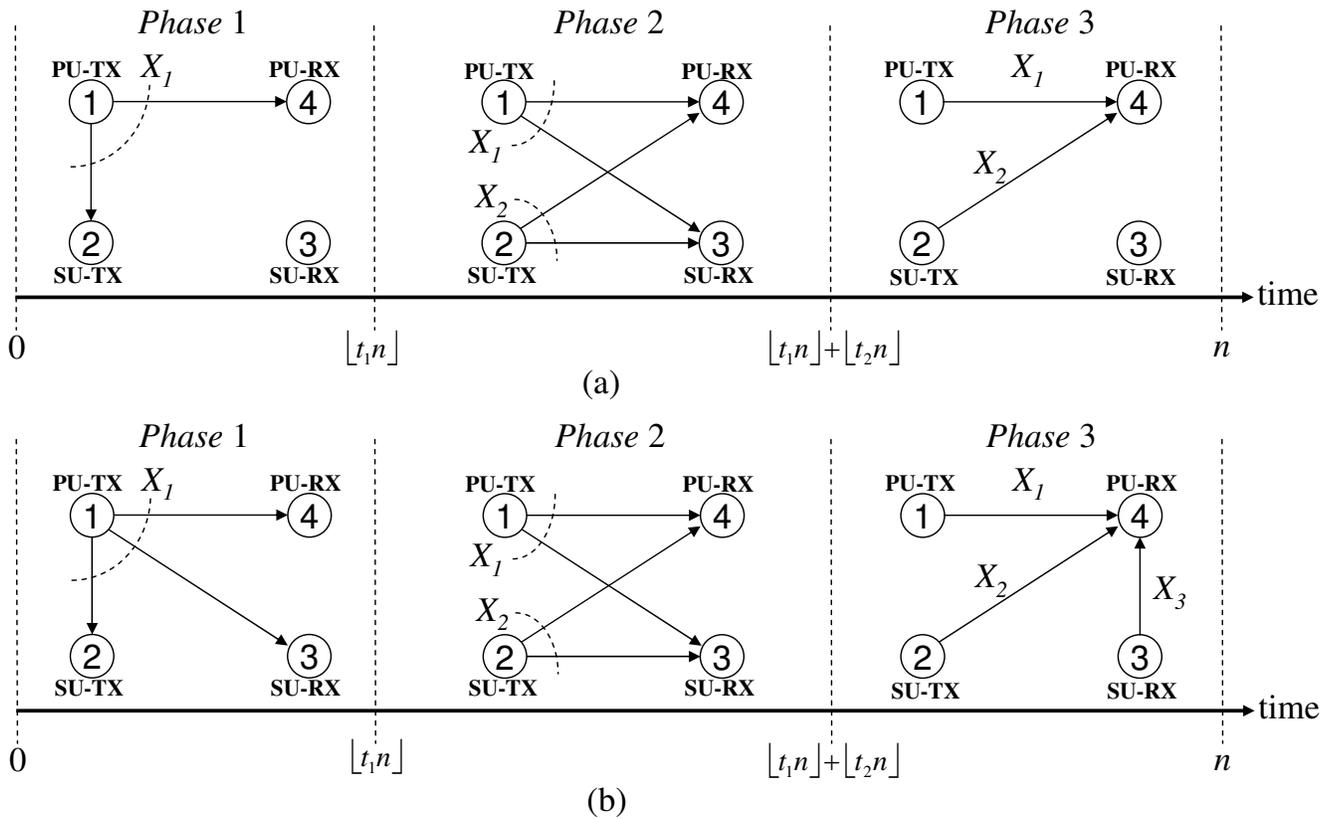,width=\textwidth}
\caption{The signaling methods of the (a) CT and (b) CTR aided
cognitive radio.} \label{Fig_UTR}
\end{figure}

\begin{figure}
\centering \epsfig{file=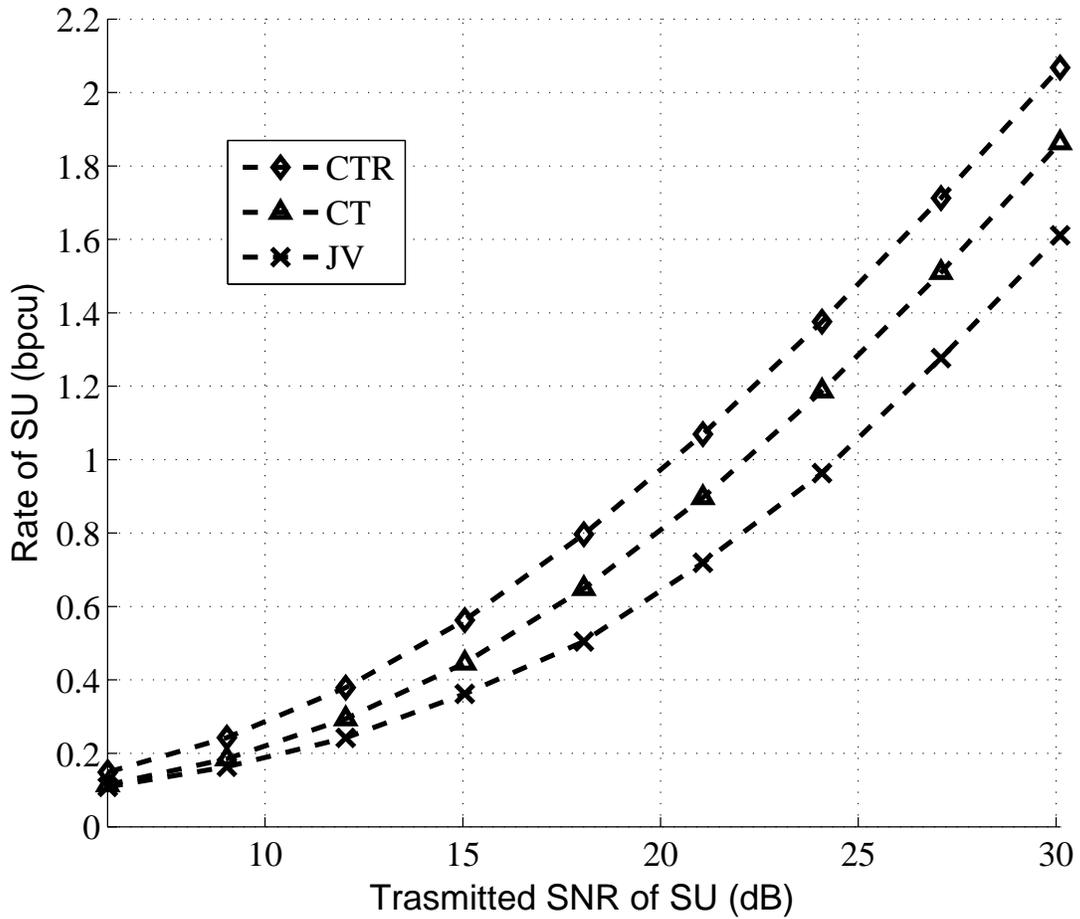 ,
width=0.9\textwidth} \caption{Comparison of the rate performance
of the SU with full CSIT, under the coexistence constraint, and
channels with large $|h_{34}|$ as specified in Table
\ref{Table_Channel_Full}. The rate is measured in bit per channel
use (bpcu).} \label{Fig_rate_SNR}
\end{figure}

\begin{figure}
\centering \epsfig{file=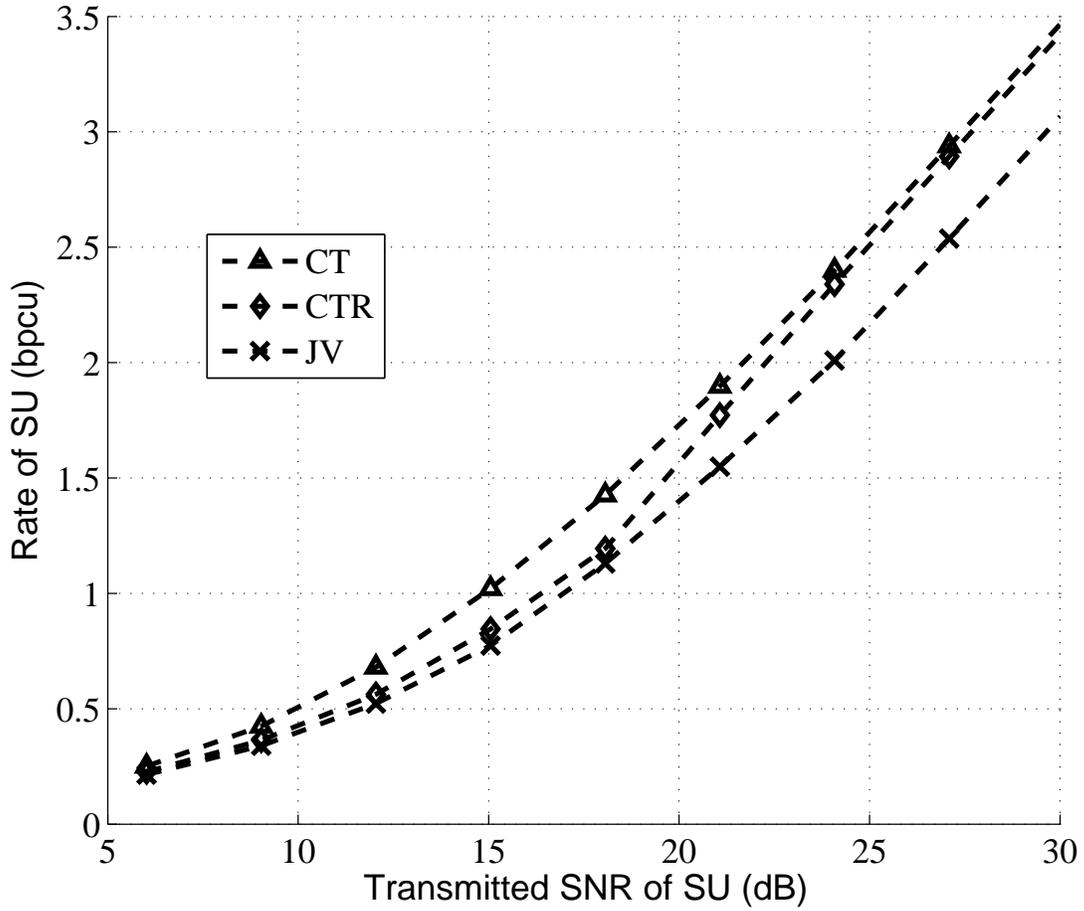
,width=0.9\textwidth} \caption{Comparison of the rate performance
of the SU with full CSIT, under the coexistence constraint, and
channels with the $|h_{34}|$ smaller than the $|h_{24}|$ as
specified in Table \ref{Table_Channel_Full}. } \label{Fig_UT_good}
\end{figure}

\begin{figure}
\centering \epsfig{file=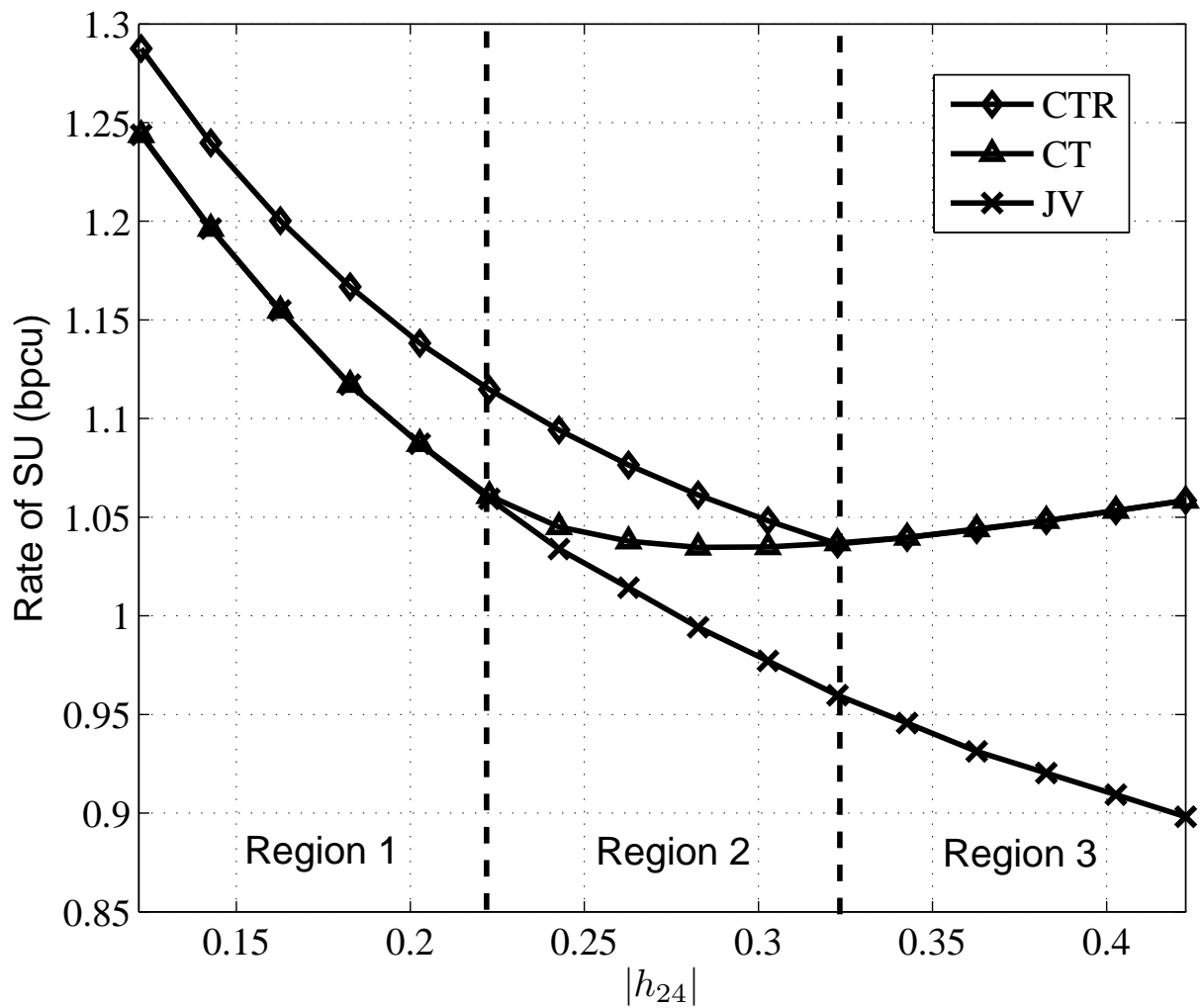 ,
width=1\textwidth} \caption{Comparison of the rate performance of
the SU with full CSIT for different $|h_{24}|$, with 20 dB
transmitted SNR and channel gains specified in Table
\ref{Table_Channel_Full}. Region 1 is the one where
$|h_{24}|<|h_{14}|$, Region 2 is the one where
$|h_{14}|<|h_{24}|<|h_{34}|$, and Region 3 is the one where
$|h_{34}|<|h_{24}|$, respectively. } \label{Fig_rate_h24}
\end{figure}

\begin{figure}
\centering \epsfig{file=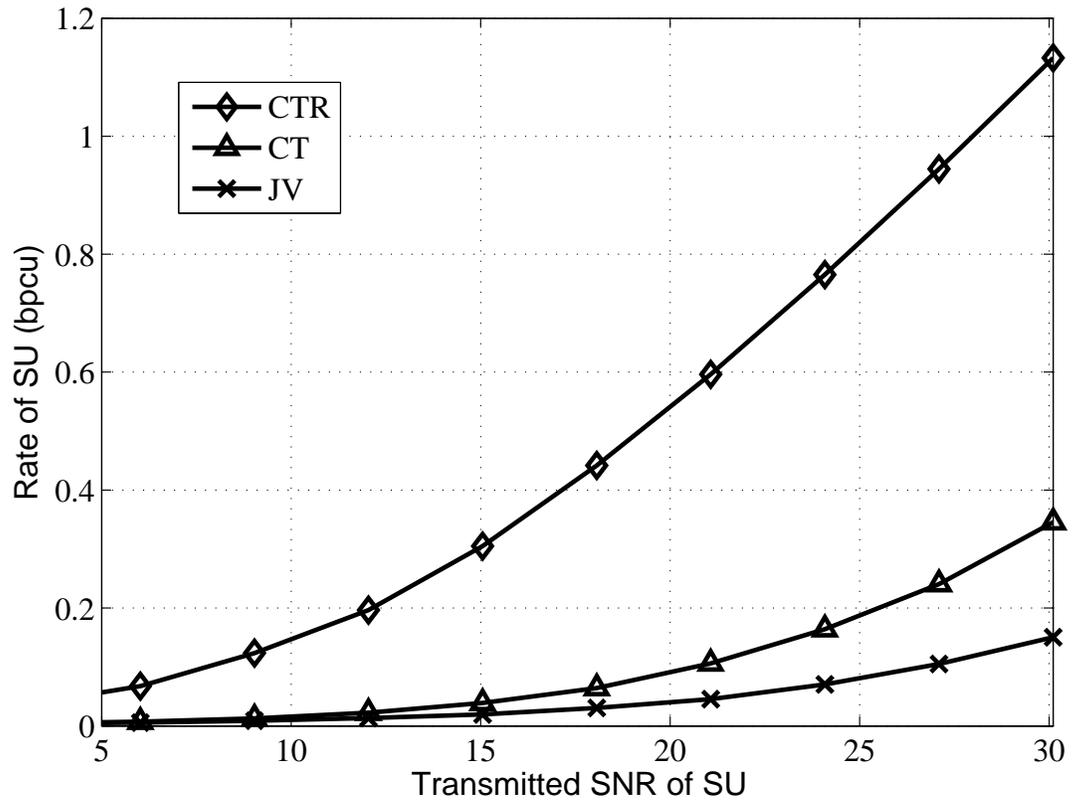 ,
width=0.9\textwidth} \caption{Comparison of the rate performance
of the SU under the coexistence constraint, and fast Rayleigh
fading channels with the statics of CSIT. The channel variances
are listed in Table \ref{Table_Channel_fading}.}
\label{Fig_ergodic_h24_leq_h34}
\end{figure}

\begin{figure}
\centering \epsfig{file=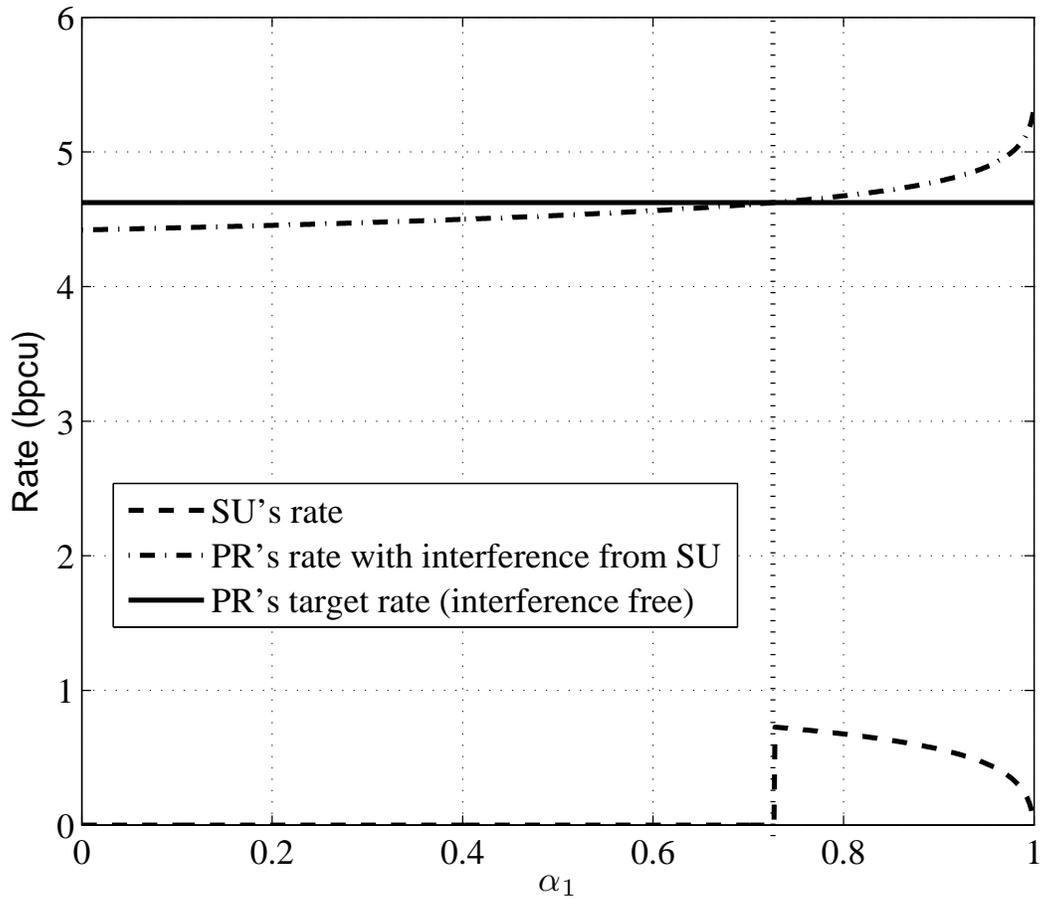
,width=0.9\textwidth} \caption{Rate performance of the SU of CTR
versus the common message relaying ratio $\alpha_1$, under fast
Rayleigh fading channels with the statics of CSIT. The transmit
SNR of SU is 20 dB and the channel variances are listed in Table
\ref{Table_Channel_fading}.} \label{Fig_coro2}
\end{figure}

\begin{figure}
\centering \epsfig{file=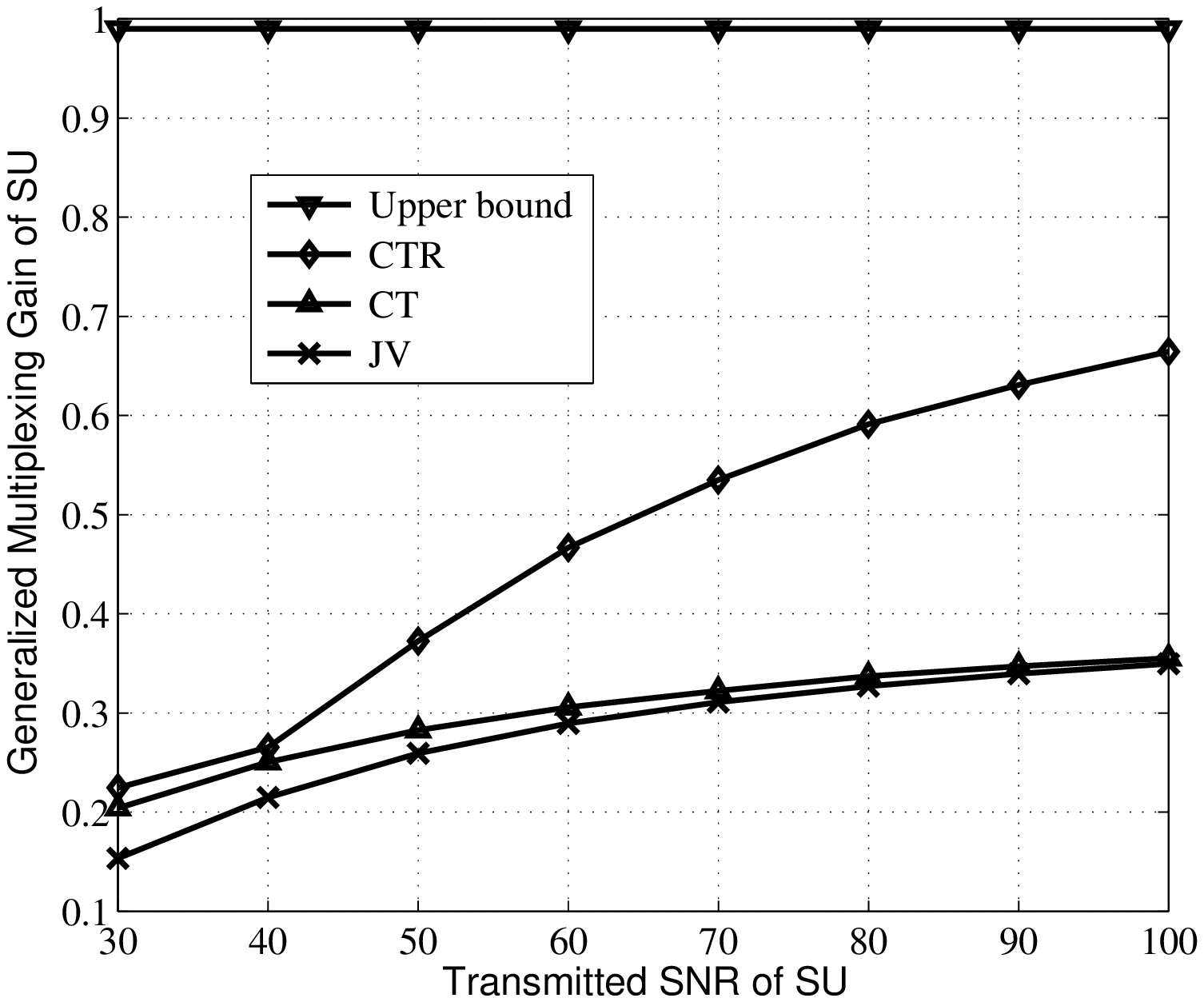 ,
width=0.9\textwidth} \caption{Comparison of the generalized
multiplexing gain performance of the SU with full CSIT, under
channels with large $|h_{34}|$ as specified in Table
\ref{Table_Channel_Full}. The upper bound is computed by
\eqref{eq_mul_UTR} with $m=0.99$.} \label{Fig_nonfading_MG}
\end{figure}

\begin{figure}
\centering \epsfig{file=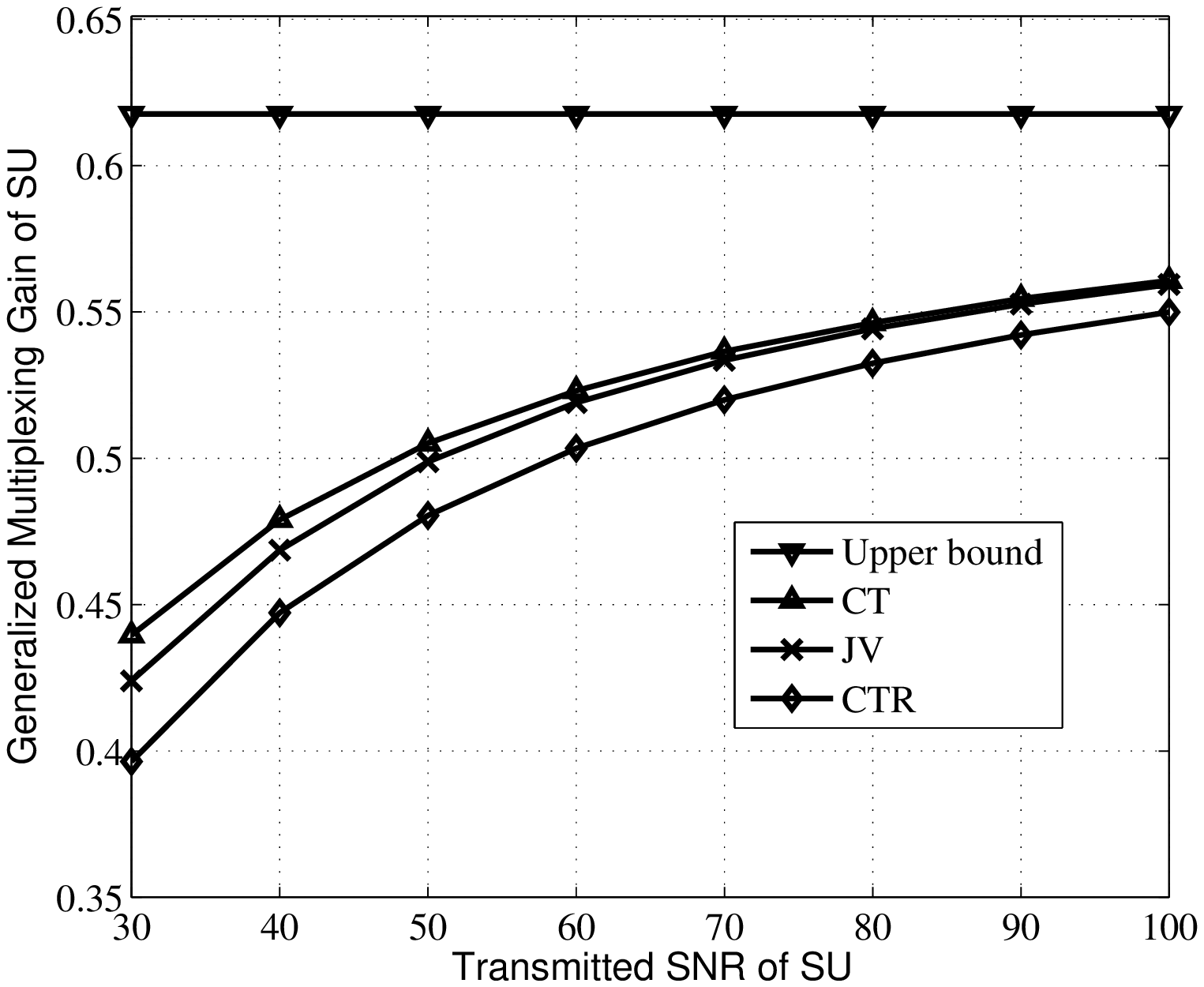
,width=0.9\textwidth} \caption{Comparison of the generalized
multiplexing gain performance of the SU with full CSIT, under
channels with small $|h_{34}|$ as specified in Table
\ref{Table_Channel_Full}. The upper bound is computed by
\eqref{eq_mul_UTR} with $m=0.99$.} \label{Fig_MG2}
\end{figure}

\begin{figure}
\centering \epsfig{file=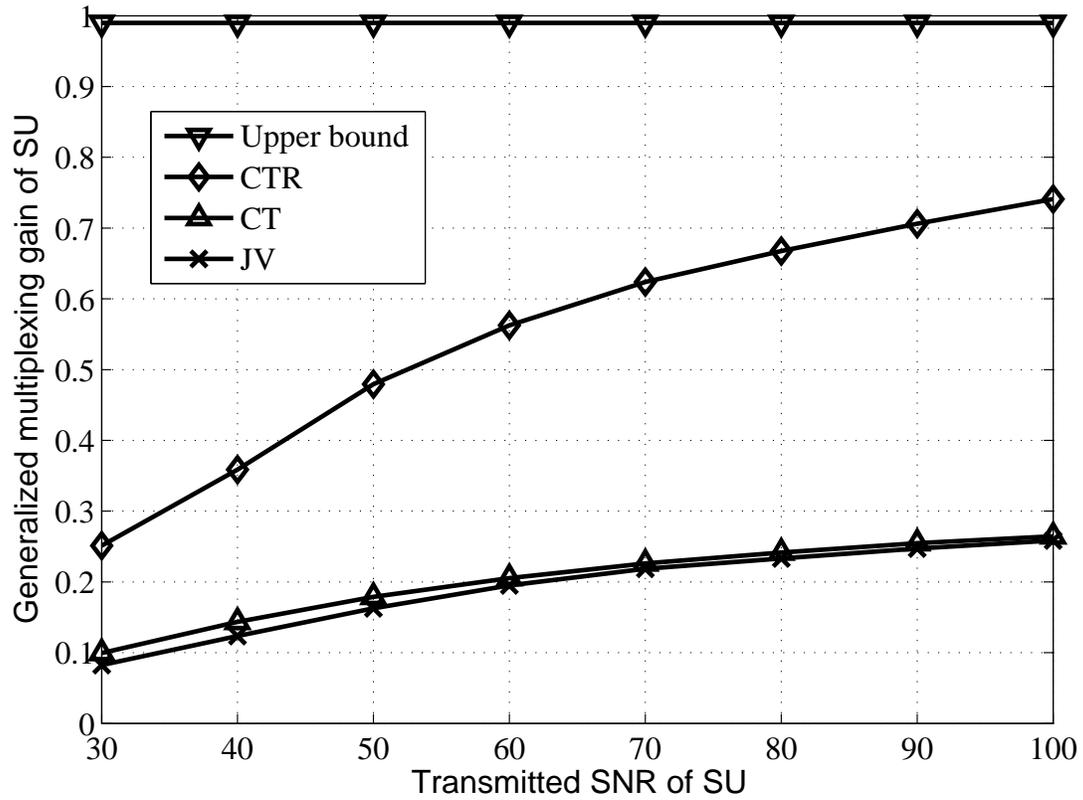
,width=0.9\textwidth} \caption{Comparison of the generalized
multiplexing gain performance of the SU, under fast Rayleigh
fading channels with the statistics of CSIT. The channel variances
are listed in Table \ref{Table_Channel_fading}. The upper bound is
computed by \eqref{eq_mul_UTR_fading} with $m=0.99$.}
\label{Fig_fading_MG1}
\end{figure}

\end{document}